\documentclass[11pt,oneside]{article}
\usepackage[top=1in, bottom=1.2in, left=1.1in, right=1.1in]{geometry}
\usepackage{setspace}
\usepackage{amsmath, amsthm, amssymb,amsbsy,amsfonts,mathrsfs}
\usepackage{color}
\usepackage{bm}
\usepackage{natbib}
\usepackage{appendix}
\usepackage{float}
\usepackage[caption = false]{subfig}
\usepackage{graphicx}
\usepackage{enumitem}
\usepackage{enumitem}
\usepackage{authblk}
\usepackage[breaklinks=true]{hyperref}
\usepackage{booktabs}
\usepackage[flushleft]{threeparttable}

\csname @openrightfalse\endcsname

\setlist{nolistsep}
\makeatletter
\newcommand*{\rom}
[1]{\expandafter\@slowromancap\romannumeral #1@}
\makeatother

\DeclareMathOperator*{\argmin}{arg\min}

\DeclareMathOperator*{\defeq}{ \overset{\text{def}}{=}}
\newtheorem{thm}{Theorem}

\newtheorem{lemma}{Lemma}

\setenumerate{noitemsep}

\hypersetup{
     colorlinks   = true,
     citecolor    =  blue
}

\usepackage{titlesec}

\newcommand{\bS}{\bar{S}}
\newcommand{\Ex}{\mathbb{E}}

\newcommand{\fx}{\mathsf{f}}

\begin{document}

\title{\textbf{Nonparametric Quantile Regressions for Panel Data Models with Large $T$} }
\author{Liang Chen\thanks{Email: chen.liang@mail.shufe.edu.cn. Financial support from the National Natural Science Foundation
of China (Grant No. 71703089) is gratefully acknowledged.}}

\affil{School of Economics, Shanghai University of Finance and Economics}

\maketitle

\begin{abstract}
This paper considers panel data models where the conditional quantiles of the dependent variables are additively separable as unknown functions of the regressors and the individual effects. We propose two estimators of the quantile partial effects while controlling for the individual heterogeneity. The first estimator is based on local linear quantile regressions, and the second is based on local linear smoothed quantile regressions, both of which are easy to compute in practice. Within the large $T$ framework, we provide sufficient conditions under which the two estimators are shown to be asymptotically normally distributed. In particular, for the first estimator, it is shown that $N\ll T^{\frac{2}{d+4}}$ is needed to ignore the incidental-parameter biases, where $d$ is the dimension of the regressors. For the second estimator, we are able to derive the analytical expression of the asymptotic biases under the assumption that $N\asymp Th^d$, where $h$ is the bandwidth parameter in local linear approximations. Our theoretical results provide the basis of using split-panel jackknife for bias corrections. A Monte Carlo simulation shows that the proposed estimators and the bias-correction method perform well in finite samples.

\vspace{0.2cm}
\noindent\textbf{Keywords}: Panel data, quantile regressions, incidental parameters, jackknife, bias correction.

\noindent\textbf{JEL codes}: C14, C31, C33.
\end{abstract}

\newpage
\section{Introduction}
This paper studies the estimation of nonparametric quantile panel data models. To facilitate the discussion, consider the following model:
\begin{equation}\label{model1}
Y_{it} =Q(X_{it}, \alpha_i, \epsilon_{it}), \text{ for } i=1,\ldots,N; t=1,\ldots,T,
\end{equation}
where $Y_{it}\in\mathbb{R}$ is the observed dependent variable, $X_{it}\in \mathcal{X} \subset \mathbb{R}^d$ is the observed regressors, $\alpha_i\in\mathbb{R}$ is the unobserved individual effect representing individual heterogeneity, and $\epsilon_{it}| (X_{it},\alpha_i)\sim \mathcal{U}(0,1) $. Similar models have also been studied by \citet{altonji2005cross} and \citet{chernozhukov2013average} under different assumptions. Assuming that the mapping $\tau \mapsto Q(x, a, \tau)$ is strictly increasing for almost all $(x,a)$ in the support of $(X_{it},\alpha_i)$, then almost surely,
\[ \mathsf{Q}_{Y_{it}}[\tau|X_{it}=x,\alpha_i=a] = Q(x, a, \tau) \overset{\text{def}}{=}  Q_{\tau}(x,a),\]
where $\mathsf{Q}_{Y_{it}}[\tau|\cdot]$ denotes the $\tau$-th conditional quantile of $Y_{it}$.
Our main object of interest is the \textit{quantile partial effects} (QPE, hereafter)  of $X_{it}$ on $Y_{it}$ while controlling for the individual effects, i.e., $\partial Q_{\tau}(x,a)/\partial x$ for $\tau\in(0,1)$.

Recent development in the literature of quantile panel data models with large $T$, including \citet{koenker2004quantile}, \citet{lamarche2010robust}, \citet{galvao2010penalized}, \citet{galvao2011quantile}, \citet{canay2011simple}, \citet{kato2012asymptotics} and \citet{galvao2016smoothed}, has mainly focused on the linear models where $Q_{\tau}(x,a) = \beta(\tau)'x + \lambda_\tau (a)$. This linearity specification for $Q_{\tau}(x,a)$ is convenient for constructing estimators of the QPE based on quantile regressions and analyzing their asymptotic properties, but it entails two possibly strong restrictions. First, in these models, $\partial Q_{\tau}(x,a)/\partial x =\beta(\tau)$, i.e., the QPE is homogeneous across $x$ and $a$. Second, the linearity assumption on $Q_{\tau}(x,a)$ usually impose strong restrictions on the regressors. For example, consider location-scale shifting models: $Y_{it} =\beta' X_{it} + \alpha_i +  g(X_{it}) \cdot \epsilon_{it}$, where $\epsilon_{it}$ is independent of $(X_{it},\alpha_i)$. In order to have $Q_{\tau}(x,a)$ linear in $x$ for all $\tau$, we need $g(x) = \gamma' x >0 $ for some $\gamma \in \mathbb{R}^d$ and almost all $x$ in the support of $X_{it}$. Thus, for $d=1$, $X_{it}$ must be positive almost surely if $\gamma>0$.

To overcome the limitations of the linearity assumption, in this paper, we consider the following more general specification:
\begin{equation}\label{model2}
Q_{\tau}(x,a) = q_{\tau}(x) + \lambda_\tau (a),
\end{equation}
which is a separable \textit{nonparametric} model in the sense that $q_{\tau}$ and $\lambda_\tau$ are both unknown functions. In this case,
\[ \partial Q_{\tau}(x,a)/\partial x = \partial q_{\tau}(x)/\partial x  \overset{\text{def}}{=}  \beta_\tau(x).\]
Thus, the QPE is allowed to be heterogeneous across $x$. Two estimators of $\beta_\tau(x)$ are proposed. The first one is based on \textit{local linear quantile regressions} (LLQR, hereafter) and the second one is based on \textit{local linear smoothed quantile regressions} (LLSQR, hereafter). The main advantage of the proposed estimators is that computationally, they are as efficient as the estimators of \citet{kato2012asymptotics} and \citet{galvao2016smoothed} for linear quantile panel models.

Despite being computationally simple, analyzing the asymptotic properties of the LLQR estimator and the LLSQR estimator in the large $T$ framework is a nontrivial task, mainly due to the well-known problem of ``incidental parameters"  --- see and \citet{lancaster2000incidental}, \citet{hahn2004jackknife} and \citet{fernandez2018fixed}. Another major contribution of this paper is that it provides a set of regularity conditions under which the proposed estimators are shown to be asymptotically normally distributed. In particular, for the LLQR estimator, the incidental-parameter biases are hard to characterize (see the discussions of \citealt{kato2012asymptotics}) and we need $N\ll T^{\frac{2}{d+4}}$ to ignore the asymptotic biases. On the other hand, under the assumption that $N\asymp Th^d$ ($h$ is the bandwidth parameter in the local linear regression), we are able to derive the asymptotic bias of the LLSQR estimator for the boundary points of $\mathcal{X}$. Interestingly, the LLSQR estimator for the interior points of $\mathcal{X}$ are shown to be free of asymptotic bias. Moreover, our asymptotic analysis provides the theoretical basis of using split-panel jackknife (see \citealt{dhaene2015split}) for bias corrections.

\vspace{0.3cm}
\noindent \textbf{Other Related Literature}

As pointed out by \citet{arellano2011nonlinear}, the identification of nonlinear panel data models with fixed $T$ is a nontrivial problem. Similarly, in the ``small $T$'' framework, the identification of $\partial Q_{\tau}(x,a)/\partial x$ is not straightforward. Invoking the result of \citet{hu2008instrumental}, one can show that for $T=3$, if $\epsilon_{i1}, \epsilon_{i2}$ and $\epsilon_{i3}$ are mutually independent conditional on $X_i\overset{\text{def}}{=}  (X_{i1},\ldots,X_{iT})'$ and some other high level conditions are satisfied, the general model \eqref{model1} is nonparametrically identified, i.e., all the conditional densities $\fx_{Y_{it}|X_i,\alpha_i}$ for $t=1,2,3$ and $\fx_{\alpha_i|X_i}$ are identified  (see Proposition 2.1 of \citealt{arellano2016nonlinear}). Given this result, the identification of $\partial Q_{\tau}(x,a)/\partial x$ follows easily. \citet{evdokimov2010identification} considers a separable model where $Q(X_{it}, \alpha_i, \epsilon_{it}) =m(X_{it}, \alpha_i) +U_{it}$ and $U_{it} \overset{\text{def}}{=}U(X_{it},\epsilon_{it} )$. In this model, $Q_{\tau}(x,a) = m(x,a)+  \mathsf{Q}_{U_{it}}[\tau|X_{it}=x] $. For $T=2$, \citet{evdokimov2010identification} provides sufficient conditions for the identification of $m(x,a)$ and $\fx_{U_{it} |X_{it}}$, which implies the identification of $\partial Q_{\tau}(x,a)/\partial x$. \citet{yan2018nonparametric} considers a similar model with $Q(X_{it}, \alpha_i, \epsilon_{it}) =m(X_{it})+\alpha_i +\sigma(X_{it}) \epsilon_{it} $. They propose a multiple-step estimator of the conditional quantile function: $m(x)+\sigma(x) \mathsf{Q}_{\epsilon}(\tau) $, but no asymptotic theory was provided for this estimator. Moreover, varying-coefficients quantile panel models where $Q_{\tau}(x,a) = \beta_\tau(x_2)'x_1+a$ and $x=(x_1',x_2')'$ is studied by \citet{su2016sieve} and \citet{cai2018semiparametric}.

The identification of the \textit{quantile treatment effects} (QTE) in nonseparable panels with fixed $T$ is considered by \citet{chernozhukov2013average} and \citet{chernozhukov2015nonparametric}. Note that the QTE considered in these papers is the derivative of the \textit{quantile structural function}: $Q^{\ast}_{\tau}(x)$, which is defined by $P[Y_{it} \leq Q^{\ast}_{\tau}(x)|X_{it}=x] =\tau $. Therefore, it is obvious that $Q^{\ast}_{\tau}(x)\neq Q_{\tau}(x,a)$, and the QTE is different from the QPE. More recently, \citet{graham2018quantile} considers the case where $ Q^{\ast}_{\tau}(x) = \beta_{\tau}(x)'x$ and focuses on the identification and estimation of the average conditional quantile effects (ACQEs) defined as $\Ex[ \beta_\tau(X_{it})]$.

Last but not least, this paper extends a large literature on nonparametric quantile regressions (see \citealt{chaudhuri1991nonparametric}, \citealt{fan1994robust}, \citealt{yu1998local}, \citealt{honda2000nonparametric}, \citealt{su2009nonparametric}, \citealt{qu2015nonparametric}, etc.) to panel data models with fixed effects.

\vspace{0.3cm}
\noindent \textbf{Structure of the Paper}

The rest of the paper is organized as follows: Section 2 introduces the models and provides some illustrative examples. Section 3 defines the estimators, whose asymptotic properties are established in Section 4. In Section 5, A Monte Carlo simulation is used to evaluate the performance of the proposed estimators and the bias-correction method. Finally, Section 6 concludes. All the proofs are collected in the appendix.

\section{The Model and Some Examples}
Specification \eqref{model2} implies the following panel data model:
\begin{equation} Y_{it} = q_{\tau}(X_{it}) + \lambda_{\tau}(\alpha_i)+ u_{it}(\tau),\end{equation}
where
the error terms satisfy the following quantile restrictions:
\[ P[ u_{it}(\tau)\leq 0|X_{it},\alpha_i] = \tau.\]
It follows that the conditional quantile of the outcomes $Y_{it}$ given the observed covariates $X_{it}$ and the individual effect $\alpha_i$ can be written as
\[ \mathsf{Q}_{Y_{it}}[\tau|X_{it}=x,\alpha_i=a] =Q_{\tau}(x,a)= q_{\tau}(x) + \lambda_{\tau}(a),\]
and as discussed in the introduction, our main object of interest is the QPE: $\beta_{\tau}(x)  = \dot{q}_{\tau}(x)$\footnote{To simply the notations we use $\dot{q}_{\tau}(x)$ and $\ddot{q}_{\tau}(x)$ to denote the first and second order derivatives of $q_\tau(\cdot)$ respectively.} for all $x\in\mathcal{X}\subset\mathbb{R}^d$ and all $\tau\in\mathcal{T}$, where $\mathcal{T}$ is a compact subset of $[0,1]$.

Consider the following 3 examples:
\begin{eqnarray*}
 \text{Example 1: }&&Y_{it} = \beta X_{it} + \alpha_i + \sqrt{1+\gamma X_{it}^2} \cdot \epsilon_{it}   \\
 \text{Example 2: }&&Y_{it} = \beta X_{it} +  \alpha_i + \Big( \sqrt{1+\gamma X_{it}^2} + \sqrt{1+\theta  \alpha_i^2} \Big)\cdot \epsilon_{it}  \\
 \text{Example 3: }&&Y_{it} = \beta X_{it} + \alpha_i +X_{it} \alpha_i+ \sqrt{1+\gamma X_{it}^2} \cdot \epsilon_{it}
\end{eqnarray*}
where $\epsilon_{it}$ is independent of $X_{it},\alpha_i $ with quantile function $\mathsf{Q}_{\epsilon}$.
It follows that
\begin{eqnarray*}
 \text{Example 1: }&&Q_{\tau}(x,a) = \underbrace{ \beta x + \sqrt{1+\gamma x^2} \cdot \mathsf{Q}_{\epsilon}(\tau)}_{q_{\tau}(x)}+ \underbrace{a}_{\lambda_{\tau}(a)}.   \\
 \text{Example 2: }&&Q_{\tau}(x,a) = \underbrace{ \beta x + \sqrt{1+\gamma x^2} \cdot \mathsf{Q}_{\epsilon}(\tau) }_{q_{\tau}(x)}+ \underbrace{ a  +\Big( \sqrt{1+\theta a ^2} \Big)\cdot \mathsf{Q}_{\epsilon}(\tau)}_{ \lambda_{\tau}(a) } . \\
 \text{Example 3: }&&Q_{\tau}(x,a)  = \beta x + a +xa+ \sqrt{1+\gamma x^2} \cdot  \mathsf{Q}_{\epsilon}(\tau).
\end{eqnarray*}
 Note that both Example 1 and Example 2 are nested by our model. In particular, the function $\lambda_{\tau}(\cdot)$ in Example 1 are invariant across $\tau \in (0,1)$. Example 3 is not nested by our model, since the conditional quantile function is not additive separable as functions of $x$ and $a$. Our model implies that the QPE is a function of $X_{it}$ only, while in Example 3 the QPE depends on both $X_{it}$ and $\alpha_i$.


\section{The Estimators}
Suppose that we have a random sample of $(Y_{it},X_{it})$ for $i=1,\ldots, N$ and $t=1,\ldots,T$, where the realized values of the individual effects are $(\alpha_{01},\ldots,\alpha_{0N})$. We follow a fixed effects approach, treating $(\lambda_{01,\tau},\ldots,\lambda_{0N,\tau}) \defeq (\lambda_\tau(\alpha_{01}),\ldots,\lambda_\tau(\alpha_{0N}))$ as fixed parameters, and consider the asymptotic framework where both dimensions of the panel data go to infinity, i.e., $N,T\rightarrow \infty$.

Focus on a single point $x \in \mathcal{X}$. Expanding $q_\tau(X_{it})$ around $x$, we have
\[ q_{\tau}(X_{it}) = q_{\tau}(x) +\dot{q}_{\tau}(x)' (X_{it}-x)  + 0.5 (X_{it}-x)' \ddot{q}_{\tau}(x)(X_{it}-x)  + R_{\tau}(x,X_{it})\]
where $R_{\tau}(x,X_{it})$ is the remainder term. It follows that
\[ Y_{it} = \lambda_{0i,\tau} + q_{\tau}(x) +\dot{q}_{\tau}(x)' (X_{it}-x)  + 0.5 (X_{it}-x)' \ddot{q}_{\tau}(x)(X_{it}-x) + R_{\tau}(x,X_{it}) +u_{it}(\tau),  \]
The above representation motivates that following LLQR estimator for $\eta_{0i,\tau}(x) = \lambda_{0i,\tau} + q_\tau(x)$ and $\beta_\tau(x)$:
\[  (\hat{\eta}_{1,\tau}(x),\ldots,\hat{\eta}_{N,\tau}(x), \hat{\beta}_{\tau}(x))=\argmin_{\eta_1,\ldots,\eta_N,\beta} \sum_{i=1}^{N}\sum_{t=1}^{T}\rho_{\tau}(Y_{it} - \eta_i - (X_{it}-x)' \beta)\cdot K\bigg(\frac{X_{it}-x}{h}\bigg)   \]
where $\rho_\tau(u)=(\tau-\mathbf{1}(u\leq 0))u$ is the check function, $K(\cdot)$ is a multivariate kernel function, and $h$ is a bandwidth parameter.
Note that
\begin{eqnarray*}
&&\rho_{\tau}(Y_{it} - \eta_i - (X_{it}-x)' \beta)\cdot K(( X_{it}-x)/h)\\
 &=& [ \tau - \mathbf{1}(Y_{it} \leq \eta_i + (X_{it}-x)' \beta)]\cdot(Y_{it} - \eta_i - (X_{it}-x)' \beta)\cdot K(( X_{it}-x)/h)\\
 &=& [ \tau - \mathbf{1}(\tilde{Y}_{it} \leq \eta_i K_{it} + \tilde{X}_{it}' \beta)] \cdot (\tilde{Y}_{it} - \eta_i K_{it} - \tilde{X}_{it}' \beta) \\
 & = &\rho_{\tau} (\tilde{Y}_{it} - \eta_i K_{it} - \tilde{X}_{it}' \beta ),
\end{eqnarray*}
where $K_{it} =K(( X_{it}-x)/h)$, $\tilde{Y}_{it} = Y_{it}  K_{it}$, and $\tilde{X}_{it} =(X_{it}-x)K_{it} $.
Thus, the LLQR estimator can be easily calculated by running a standard quantile regression of $\tilde{Y}_{it}$ on $\tilde{X}_{it}$ and $N$ additional regressors: $\mathbf{1}(i=1)K_{it}, \ldots, \mathbf{1}(i=N)K_{it}$, therefore it is very computationally efficient.

Inspired by \citet{galvao2016smoothed}, we also consider the following LLSQR estimator:
\begin{multline*}
 (\check{\eta}_{1,\tau}(x),\ldots,\check{\eta}_{N,\tau}(x), \check{\beta}_{\tau}(x))\\ =\argmin_{\eta_1,\ldots,\eta_N ,\beta }  \sum_{i=1}^{N}\sum_{t=1}^{T}\Bigg(\tau - G\Bigg(\frac{Y_{it} - \eta_i - (X_{it}-x)' \beta }{b} \Bigg) \Bigg)(Y_{it} - \eta_i - (X_{it}-x)' \beta)\cdot K\bigg(\frac{X_{it}-x}{h}\bigg),
\end{multline*}
where $G(z) = 1-\int_{-\infty}^{z}g(u)du$, $g(\cdot)$ is a continuously differentiable function with support $[-1,1]$, and $b$ is a bandwidth parameter. The idea of smoothed quantile regression (see \citealt{amemiya1982two} and \citealt{horowitz1998bootstrap}) is to approximate the non-smooth indicator function with a smooth cumulative distribution function.

\section{Asymptotic Results}
Before presenting the asymptotic results, it is useful to define some new notations. Let $x_{\partial}$ be on the boundary of $\mathcal{X}$. The \textit{boundary points} are defined as
\[ x=x_{\partial}+c h  \text{ for some } c\in \text{supp}(K),\]
and the domain for integration is defined as
\[ \mathcal{B} = \{ v\in\mathbb{R}^d: (c+v)h\in \mathcal{X} \} \cap \text{supp}(K).\] Define
\[c_0= \int_{\mathcal{B}} K(u)du, \text{ }\mathcal{C}_1= \int_{\mathcal{B}} uK(u)du, \text{ } \mathcal{C}_2= \int_{\mathcal{B}} uu'K(u)du ,\text{ }
 \mathcal{C}= \mathcal{C}_2-\mathcal{C}_1 \mathcal{C}_1'/c_0,\]
 \[  \bar{\mathcal{C}}_2 =\begin{pmatrix} c_0 &  \mathcal{C}_1' \\ \mathcal{C}_1 & \mathcal{C}_2   \end{pmatrix}, \text{ }  d_0= \int_{\mathcal{B}} K^2(u)du, \text{ }\mathcal{D}_1= \int_{\mathcal{B}} uK^2(u)du, \text{ } \mathcal{K}_1=\int uu'K(u)du , \text{ } \mathcal{K}_2=\int uu'K^2(u)du. \]

\subsection{Asymptotic Distribution of the LLQR Estimator}
Write $u_{it}$ instead of $u_{it}(\tau)$ to simply the notations. Let $B_{\epsilon}$ be a neighbourhood of $0$. We first impose the following assumptions:
\begin{enumerate}
\item[(A1)] $\mathcal{X}$ is compact.
\item[(A2)] $(Y_{it},X_{it})$ are independent of $(Y_{js},X_{js})$ for any $i\neq j$ or $t\neq s$. $(Y_{i1},X_{i2}),\ldots, (Y_{iT},X_{iT})$ are identically distributed for each $i$.
\item[(A3)] Let $f_{u,i}(\cdot|x)$ denote the conditional density of $u_{it}$ given $X_{it}=x$ and let $f_{X,i}(\cdot)$ denote the density of $X_{it}$. There exists $c_2>c_1>0$ such that $c_1<f_{u,i}(0|x)<c_2$ and $c_1<f_{X,i}(x)<c_2$ for all $i$ and all $x \in \mathcal{X}$.
\item[(A4)] Define $f_{u,i}^{(1)} (c|x) =\partial f_{u,i}(c|x)/\partial c$ and $f_{X,i}^{(1)} (x) =\partial f_{X,i}(x)/\partial x$. There exists a $M>0$ such that $| f_{u,i}^{(1)} (c|x)|<M$ for all $c \in B_{\epsilon}$ and $|f_{X,i}^{(1)} (x)|, | \partial f_{u,i}(c|x)/\partial x|<M$ for all $x\in\mathcal{X}$. Moreover, $| \partial^2 q_{\tau}(x)/ \partial x_j \partial x_p| <M$, $|\partial^3 q_{\tau}(x)/ \partial x_j \partial x_p \partial x_h|<M$ for all $j,p,h \leq d$.
\item[(A5)] The kernel function $K$ has bounded support and $\int uK(u)du=0$, $\int u_j u_p u_h K(u)du=0$ for all $j,p,h \leq d$.
\item[(A6)] $c_0>0$, $\mathcal{K}_1>0$, $\bar{\mathcal{C}}_2>0$ and $\mathcal{C}>0$, and that
\[0< \sigma(x)  \defeq \frac{ \lim_{N\rightarrow\infty} N^{-1}\sum_{i=1}^{N}f_{X,i}(x) }{ \lim_{N\rightarrow\infty} [N^{-1}\sum_{i=1}^{N} f_{X,i}(x)  f_{u,i}(0|x)]^2 }  <\infty  \]
for all $x\in\mathcal{X}$.
\item[(A7)] Let $N \asymp T^{c_N}$ and $h \asymp T^{-c_h}$ for some $c_N,c_h>0$. Then $c_N< \frac{2}{d+4}$ and $\frac{1}{d+4}<c_h<\frac{1-2c_N}{d}$.
\end{enumerate}

\vspace{0.3cm}
\noindent{\textbf{Remark 1.1:}} The above assumptions, except (A2) and (A7), are standard in the literature of local linear quantile regressions and quantile regressions. Note that we only need the existence and smoothness of the conditional density of $u_{it}$ given $X_{it}$, thus the estimator is robust to heavy tails and outliers in $u_{it}$.

\vspace{0.3cm}
\noindent{\textbf{Remark 1.2:}} The independence assumption (A2) is also adopted by \citet{kato2012asymptotics} and it excludes time-invariant regressors. This independence assumption can be relaxed to allow for $\beta$-mixing on the time dimension along the line of \citet{galvao2016smoothed} at the cost of much lengthier proofs. Thus, to keep the proofs tractable, (A2) is maintained throughout the paper.

\vspace{0.3cm}
\noindent{\textbf{Remark 1.3:}} Assumption (A7) ensures that $\log N \ll Th^{d+4}$, $N\ll Th^{d+2}$, $NTh^{d+6}\rightarrow 0$, and $N^2 \ll Th^d$. These conditions are needed to prove Theorem 1 below. For example, for $d=2$ and $c_N=1/4$, we can choose $c_h=1/5$. Note that due to the nonparametric nature of our estimator, the condition $N^2 \ll Th^d$ imposed here is stronger than the condition $N^2 \ll T$ required by \citet{kato2012asymptotics}, since the order of the incidental-parameter bias is approximately $(Th^d)^{-3/4}$, while in \citet{kato2012asymptotics} the bias is approximately of order $T^{-3/4}$. Such conditions are hard to justify in practice, this is why we also consider the LLSQR estimator, whose asymptotic distribution can be established under more realistic assumptions about the relative sizes of $N$ and $T$.

\vspace{0.3cm}
The following theorem gives the asymptotic distribution of the LLQR estimator.
\begin{thm} Suppose that Assumptions (A1) to (A7) hold, then:

(i) For any interior point $x$ of $\mathcal{X}$, we have
\[ \sqrt{NTh^{d+2}} \big[ \hat{\beta}_{\tau}(x) - \beta_{\tau}(x)  \big] \overset{d}{\rightarrow} \mathcal{N}\Big(0, \tau(1-\tau) \sigma(x)  \mathcal{K}_1^{-1}\mathcal{K}_2 \mathcal{K}_1^{-1} \Big) .\]
(ii) For any boundary point $x$ of $\mathcal{X}$, i.e., $x=ch$ for some $c>0$ in the support of $K(\cdot)$, we have
\[ \sqrt{NTh^{d+2}} \big[ \hat{\beta}_{\tau}(x) - \beta_{\tau}(x)  - h B^{(1)}\big] \overset{d}{\rightarrow} \mathcal{N}\Big(0, \tau(1-\tau) \sigma(0) \Omega  \Big), \]
where
\[\Omega  =  \mathcal{C}^{-1} \left[ \int_{\mathcal{B}}(u-\mathcal{C}_1/c_0)(u-\mathcal{C}_1/c_0)' K^2(u)du \right] \mathcal{C}^{-1} \text{ and }B^{(1)}=  0.5\mathcal{C}^{-1}\int_{\mathcal{B}} u' \ddot{q}_{\tau}(0)u \left(u-\mathcal{C}_1/c_0\right)K(u)du.\]
\end{thm}

\vspace{0.3cm}
\noindent{\textbf{Remark 1.4:}} The LLQR estimator suffers from two types of biases: a bias due to the estimation of incidental parameters, and another one due to local linear approximations. As discussed in Remark 1.3, the first bias can be ignored at the expense of a very strong condition: $N\ll T^{\frac{2}{d+4}} $. The term $h B^{(1)}$ is the leading bias term in the local linear approximations (see \citealt{fan1994robust} for example). This bias can be further reduced by using local polynomial regressions. Note that $B^{(1)}=0$ for the interior points, thus the leading bias term for the estimators of the interior points is $O(h^2)$.

\vspace{0.3cm}
\noindent{\textbf{Remark 1.5:}} In general, it is straightforward to construct consistent estimators of the asymptotic variances, since $\mathcal{K}_1,\mathcal{K}_2$ and $\Omega$ only depend on the kernel function $K(\cdot)$, and $\sigma(x)$ can be consistently estimated using standard nonparametric methods. In particular, if the distribution of $(u_{it},X_{it})$ are identical across $i$, i.e., $ f_{u,i}(0|x) = f_{u}(0|x)$ and $f_{X,i}(x)=f_{X}(x)$ for all $i$, then $ \sigma(x)  = \left( f_{X}(x) \right)^{-1}\cdot  \left(   f_{u}(0|x) \right)^{-2}$. 

\subsection{Asymptotic Distribution of the LLSQR Estimator}

We impose the following assumptions:
\begin{enumerate}
\item[(B1)] Assumptions (A1) to (A6) hold.
\item[(B2)] $f_{u,i}(c|x)$ is $m+2$ times continuously differentiable in $c$ for all $x\in\mathcal{X}$. Let $f_{u,i}^{(j)}(c|x) = \partial^j f_{u,i}(c|x)/ \partial c^j$, then there exists some $M$ such that $|f_{u,i}^{(j)}(c|x)| <M$ for all $j\leq m+2$.
\item[(B3)] $g(v)$ is a symmetric function with support $[-1,1]$ and $\int g(v)dv=1$. For some positive integer $m\geq 4$, $\int v^j g(v)dv=0$ for $j=1,\ldots,m-1$, and $\int v^m g(v)dv<\infty$.
\item[(B4)] $N/(Th^d)\rightarrow \kappa^2$ for some $\kappa>0$. $ h \asymp T^{-c_h}$ and $ b \asymp T^{-c_b}$ for some $c_h,c_b>0$ that satisfy
\begin{equation}\label{B4_1} c_h \in \left( \frac{1}{d+3},\frac{1}{d+1}\right) \cup \left(\frac{c_b}{2},\frac{mc_b}{2} \right) \cup \left(\frac{1-mc_b}{d},\frac{1-3c_b}{m} \right) \end{equation}
\begin{equation}\label{B4_2} \frac{2}{m(2+d)}<c_b <\frac{2}{d+6}.\end{equation}
\end{enumerate}

\vspace{0.3cm}
\noindent{\textbf{Remark 2.1:}} Assumptions (B2) and (B3) are also imposed in \citet{galvao2016smoothed}. In particular, we need $g(\cdot)$ to be a fourth (or higher) order kernel function. Assumption (B4) is new. Condition \eqref{B4_1} implies that $Th^{d+1}\rightarrow\infty$, $Th^{d+3}\rightarrow0$, $Th^db^3\rightarrow\infty$, $Th^db^m\rightarrow0$ and $ b^m \ll h^2 \ll b$. These conditions will be used in the proof of Theorem 2. Moreover, $m\geq4$ and \eqref{B4_2} ensure that $c_h$ lies in a non-empty set. For example, for $m=4$ and $d=2$, one can choose $c_b=1/6$ and $c_h\in(1/5,1/4)$.

\vspace{0.3cm}
The following theorem gives the asymptotic distribution of the LLSQR estimator.
\begin{thm} Suppose that Assumptions (B1) to (B4) hold, then:

(i) For any interior point $x$ of $\mathcal{X}$, we have
\[ \sqrt{NTh^{d+2}} \big[ \check{\beta}_{\tau}(x) - \beta_{\tau}(x)  \big] \overset{d}{\rightarrow} \mathcal{N}\Big(0, \tau(1-\tau) \sigma(x)  \mathcal{K}_1^{-1}\mathcal{K}_2 \mathcal{K}_1^{-1} \Big) .\]
(ii) For any boundary point $x$ of $\mathcal{X}$, i.e., $x=ch$ for some $c>0$ in the support of $K(\cdot)$, we have
\[ \sqrt{NTh^{d+2}} \big[ \check{\beta}_{\tau}(x) - \beta_{\tau}(x)  - h B^{(1)}\big] \overset{d}{\rightarrow} \mathcal{N}\Big(\kappa B^{(2)}, \tau(1-\tau) \sigma(0) \Omega  \Big),\]
where
\[ B^{(2)}= - \frac{\tau-1/2}{\lim_{N\rightarrow\infty} N^{-1}\sum_{i=1}^{N} f_{X,i}(x)  f_{u,i}(0|x) } \cdot \mathcal{C}^{-1}(\mathcal{D}_1/c_0 -\mathcal{C}_1 d_0/c_0^2 ).\]
\end{thm}

\vspace{0.3cm}
\noindent{\textbf{Remark 2.2:}} It can be seen that the asymptotic distributions of the LLSQR estimators and the LLQR estimators are very similar, with one noticeable difference: the LLSQR estimator for the boundary points suffers from an asymptotic bias: $\kappa B^{(2)}$, which is the consequence of estimating incidental parameters. In the proof of Theorem 2, it is found that the incidental-parameter bias of the LLSQR estimator for the boundary points is of order $(Th^d)^{-1}$ rather than $T^{-1}$ --- this is why we need $N\asymp Th^d$ to derive the analytical expression of the asymptotic bias. Interestingly, the LLSQR estimators for the boundary points at $\tau=0.5$ and the LLSQR estimators for the interior points at all $\tau$s are all free of asymptotic biases. These findings are further confirmed by a Monte Carlo simulation in Section 5.

\vspace{0.3cm}
\noindent{\textbf{Remark 2.3:}} Theorem 2 provides the theoretical basis for bias corrections using the split-panel jackknife method proposed by  \citet{dhaene2015split}. In particular, divide the whole sample into two subsamples: $(Y_{it},X_{it})$ for $i=1,\ldots,N; t=1,\ldots, T/2$, and $(Y_{it},X_{it})$ for $i=1,\ldots,N; t=T/2+1,\ldots, T$, and let $\check{\beta}_{\tau,1}(x), \check{\beta}_{\tau,2}(x)$ denote the LLSQR estimators using the two subsamples respectively\footnote{The bandwidth parameter $h$ should be the same for $\check{\beta}_{\tau}(x)$, $\check{\beta}_{\tau,1}(x)$ and $\check{\beta}_{\tau,2}(x)$.}. The bias-corrected estimator is simply given by
\[ \check{\beta}_{\tau}^{bc}(x) = 2\check{\beta}_{\tau}(x) - 0.5 [\check{\beta}_{\tau,1}(x)+\check{\beta}_{\tau,2}(x)].  \]
Under Assumptions (B1) to (B4) we can show that for the boundary points,
\[  \sqrt{NTh^{d+2}} \big[ \check{\beta}_{\tau}^{bc}(x) - \beta_{\tau}(x)  - h B^{(1)}\big] \overset{d}{\rightarrow} \mathcal{N}\Big(0, \tau(1-\tau) \sigma(0) \Omega  \Big).   \]

\vspace{0.3cm}
\noindent{\textbf{Remark 2.4:}} Assumption (B4) requires that $N\asymp Th^d$, which is much less restrictive than Assumption (A6) which imposes $N \ll \sqrt{Th^d}$. As discussed in Remark 2.1, for $d=2$, Assumption (B4) admits the choice: $c_h=1/4.5$ and therefore $N\asymp T^{5/9}$. Thus, given the nonparametric nature of the problem, Assumption (B4) is still more stringent than the usual assumption $N\asymp T$ imposed for nonlinear fixed-effects estimators (see  \citealt{hahn2004jackknife} and \citealt{fernandez2018fixed}).

\section{A Monte Carlo Simulation}
In this section, we evaluate the performance of the proposed estimators in finite samples using the following data generating process (DGP):
\[ Y_{it} = \beta X_{it} + \alpha_i + \sqrt{ 1+X_{it}^2} \cdot \epsilon_{it}, \]
where $X_{it}\sim i.i.d \text{ } \mathcal{N}(0,1)\cdot \mathbf{1}\{|X_{it}|\leq 2\}$, $\alpha_i\sim i.i.d \text{ } \mathcal{N}(0,1)$. It is easy to see that $\beta_{\tau}(x) = 1+ \mathsf{Q}_{\epsilon}(\tau)\cdot x/\sqrt{1+x^2}$, where $\epsilon_{it}$ are i.i.d with quantile function $\mathsf{Q}_{\epsilon}$. We consider two different distributions of $\epsilon_{it} $: (i) $\mathcal{N}(0,1)$, and (ii) t distribution with 3 degrees of freedom, and compare the biases and mean-square errors (MSEs) of four different estimators: the LLQR estimator $\hat{\beta}_{\tau}$, the LLSQR estimator $\check{\beta}_{\tau}$, and the bias-corrected versions of these two estimators, denoted as $\hat{\beta}_{\tau}^{bc}$ and $\check{\beta}_{\tau}^{bc}$ respectively.

To same space, we only report the results for $N=T=100$, $\tau=0.25,0.5,0.75$ and $x=-2,-1.6,\ldots,1.6,2$. For all estimators, we choose $h=0.8$, and for the LLSQR estimators, we choose $b=0.5$ and consider the following fourth-order kernel function:
\[ k(u)=\frac{105}{64}\left(1-5 u^{2}+7 u^{4}-3 u^{6}\right) 1(|u| \leq 1). \]
We have also tried other bandwidth values and find that the results is more sensitive to the choice of $h$ than the choice of $b$.

Table 1 reports the results for $\tau=0.25$ and $\epsilon_{it}\sim\mathcal{N}(0,1)$ while Table 2 reports the results for $\tau=0.25$ and $\epsilon_{it}\sim T(3)$. The results for $\tau=0.5$ and $\tau=0.75$ are reported in Table 3 to Table 6.

For $\tau=0.25$, four conclusions can be drawn from the results in Tables 1 and 2. (i) The performance of $\hat{\beta}_{\tau}$ and $\check{\beta}_{\tau}$, in terms of biases and MSEs, are very close. (ii) As predicted by our Theorem 2, the bias-correction method significantly reduces the biases of the LLSQR estimators, especially at the boundary points (e.g., $|x|=2,1.6$). Interestingly, the bias-correction method can also effectively reduce the biases of the LLQR estimators. (iii) The bias-corrected estimators for the boundary points have much lower MSEs due to the large decrease in biases. (iv) The performance of the estimators are robust to heavy tails of $\epsilon_{it}$. Similar conclusions are supported by the results for $\tau=0.75$. However, for $\tau=0.5$, the bias correction at the boundary points is not very effective --- this is predicted by Theorem 2, which shows that the LLSQR estimator for the boundary points is free of asymptotic biases at $\tau=0.5$ (see Remark 2.2).

\section{Conclusion}
To the best of our knowledge, this is the first paper that considers nonparametric quantile regressions in the context of large $T$ panels. Our model is additively separable as unknown functions of the regressors and the individual effects, and it allows the QPE to be heterogeneous across individuals. We propose two estimators of the QPE based on local linear approximations, and establish their asymptotic distributions under a set of regularity assumptions. Our theoretical results highlight the importance of incidental-parameter biases and justify the use of convenient jackknife method to correct the asymptotic biases. The good performance of the bias-correction method in finite samples is confirmed using a Monte Carlo simulation.

Like any other nonparametric estimators, the choice of bandwidth is crucial in practice. In this paper we have focused on the theoretical conditions that the bandwidth parameters have to satisfy, but how to choose those bandwidths in practice is an important question that is left for further investigation.

\appendix

\numberwithin{equation}{section}
\newpage

\section{Proofs of The Main Results}
\subsection{Proof of Theorem 1}
To simply the notations, we suppress the dependence of the parameters on $\tau$ and $x$. For example, we write $\eta_{0i},\hat{\eta}_i,\hat{\beta},\beta_0$ instead of $\eta_{0i,\tau}(x),\hat{\eta}_{i,\tau}(x), \hat{\beta}_{\tau}(x),\beta_{\tau}(x)$. Moreover, we first define the following notations: $\phi = h \beta, \hat{\phi}=h\hat{\beta}, \phi_0=h\beta_{\tau}(x)$, $\theta_i = (\eta_i, \phi')'$, $\theta = (\eta_1,\ldots, \eta_N, \phi ' )'$, $\theta_{0i} = (\eta_{0i}, \phi_0')'$, $\theta_0 = (\eta_{01},\ldots, \eta_{0N}, \phi_0')'$, $W_{it} =(1, (X_{it}-x)'/h)'$, $\hat{\theta}_i = (\hat{\eta}_i, \hat{\phi}')'$, $\hat{\theta} = (\hat{\eta}_1,\ldots, \hat{\eta}_N, \hat{\phi}')'$, $\psi_{\tau}(u)=\mathbf{1}\{u\leq 0\}-\tau$,
\[ S_{T,i}(\theta_i) = \frac{1}{T h^d}\sum_{t=1}^{T}\rho_{\tau}(Y_{it} - \theta_i'W_{it})K_{it}, \quad  \bar{S}_{i}(\theta_i) = \Ex[ \rho_{\tau}( (\theta_{0i} - \theta_i)'W_{it}+u_{it})K_{it}/h^d]\]
and $S_{NT}(\theta)  = N^{-1}\sum_{i=1}^{N}S_{T,i}(\theta_i) $, $\bS_{N}(\theta) = N^{-1}\sum_{i=1}^{N} \bar{S}_{i}(\theta_i)$.

\vspace{0.5cm}

\begin{lemma} Under Assumptions A1 to A6, we have $ \|\hat{\phi}-\phi_0\|_1 =o_P(1)$ and $\max_{i\leq N}|\hat{\eta}_i -\eta_{0i}|=o_P(1)$.
\end{lemma}
\begin{proof}
For any $\delta>0$, define $B_i(\delta) =\{ \theta_i: |\eta_i-\eta_{0i} |+ \|\phi-\phi_0\|_1\leq \delta\}$. For any $\bar{\theta}_i \in B_i^C(\delta)$, define $\tilde{\theta}_i =r_i \bar{\theta}_i +(1-r_i) \theta_0$, where $r_i= \delta/(|\bar{\eta}_i-\eta_{0i} |+\|\bar{\phi}-\phi_0\|_1)<1$. Note that $\tilde{\theta}_i$ is on the boundary of $B_i(\delta)$. By the convexity of $S_{T,i}(\theta_i)$ we have
\[   S_{T,i}(\tilde{\theta}_i)  \leq r_i S_{T,i}(\bar{\theta}_i) +(1-r_i)S_{T,i}(\theta_{0i}), \text{ or } S_{T,i}(\tilde{\theta}_i)  - S_{T,i}(\theta_{0i})  \leq r_i  (S_{T,i}(\bar{\theta}_i) -S_{T,i}(\theta_{0i}) ).  \]

Next, by the definition of the estimator, we have $S_{T,i}(\hat{\theta}_i) \leq S_{T,i}(\theta_{0i})$ for some $i\leq N$. Thus, if $\|\hat{\phi}-\phi_0\|_1 >\delta$, then $\hat{\eta}_i \in B_i^C(\delta)$, which (by the above inequality) implies that
\[  ( S_{T,i}(\tilde{\theta}_i)  - S_{T,i}(\theta_{0i}) )/r_i  \leq  S_{T,i}(\hat{\theta}_i) -S_{T,i}(\theta_{0i})\leq 0 . \]
Adding the subtracting terms, the above inequality can be written as
\[ \bar{S}_{i}(\tilde{\theta}_i) - \bar{S}_{i}(\theta_{0i}) \leq  S_{T,i}(\theta_{0i}) - \bar{S}_{i}(\theta_{0i}) - [ S_{T,i}(\tilde{\theta}_i) - \bar{S}_{i}(\tilde{\theta}_i)   ] .\]
For interior points, using Taylor expansion, we have for some $C>0$ that does not depend on $i$, and small enough $\delta$,
\begin{multline*} \bar{S}_{i}(\tilde{\theta}_i) - \bar{S}_{i}(\theta_{0i})  = \Ex[ \psi_{\tau} (u_{it})W_{it}K_{it}/h^d] (\tilde{\theta}_i -\theta_{0i} ) \\+0.5 f_{u,i}(0|x)f_{X,i}(x)  [ (\tilde{\eta}_i -\eta_{0i})^2 +(\tilde{\phi} - \phi_0 )' \mathcal{K}_1(\tilde{\phi} - \phi_0 )   ](1+o(1)) > C \delta^2 (1+o(1)),
\end{multline*}
where we have used the fact that $ \Ex[ \psi_{\tau} (u_{it})W_{it}K_{it}/h^d] =0$. Similarly, for the boundary points, we have
\[\bar{S}_{i}(\tilde{\theta}_i) - \bar{S}_{i}(\theta_{0i})  =0.5 f_{u,i}(0|x)f_{X,i}(x)  (\tilde{\theta}_i - \theta_{0i} )' \bar{\mathcal{C}}_{2}(\tilde{\theta}_i - \theta_{0i} )  (1+o(1)) > C \delta^2 (1+o(1)) .\]
Thus, if follows from the union bound that
\[ P[\|\hat{\phi}-\phi_0\|_1  >\delta ] \leq \sum_{i=1}^{N}P\big[S_{T,i}(\theta_{0i}) - \bar{S}_{i}(\theta_{0i}) -  ( S_{T,i}(\tilde{\theta}_i) - \bar{S}_{i}(\tilde{\theta}_i)) >C \delta^2 (1+o(1))  \big].\]
To prove $\|\hat{\phi}-\phi_0\|_1  =o_P(1)$, it then suffices to show that for any $\epsilon>0$
\begin{equation}\label{A1} \max_{i\leq N} P\Bigg[ \sup_{\theta_i \in B_i(\delta)}| S_{T,i}(\theta_{0i}) - \bar{S}_{i}(\theta_{0i}) -  ( S_{T,i}(\theta_i) - \bar{S}_{i}(\theta_i)) | >\epsilon  \Bigg] = o(N^{-1}). \end{equation}
Note that
\begin{eqnarray*}
&&S_{T,i}(\theta_{i}) - \bar{S}_{i}(\theta_{i}) \\
&=&\frac{1}{T h^d}\sum_{t=1}^{T}\rho_{\tau}(Y_{it} - \theta_i'W_{it})K_{it} -  \frac{1}{T h^d}\sum_{t=1}^{T}\Ex[ \rho_{\tau}( (\theta_{0i} - \theta_i)'W_{it}+u_{it})K_{it}]   \\
&=& \underbrace{ \frac{1}{T h^d}\sum_{t=1}^{T}\big[\rho_{\tau}(Y_{it} - \theta_i'W_{it})- \rho_{\tau}( (\theta_{0i} - \theta_i)'W_{it}+u_{it}) \big] K_{it}}_{A_{T,i}(\theta_i)} \\
&&+ \underbrace{\frac{1}{T h^d}\sum_{t=1}^{T} \Big[ \rho_{\tau}( (\theta_{0i} - \theta_i)'W_{it}+u_{it})K_{it} -  \Ex[ \rho_{\tau}( (\theta_{0i} - \theta_i)'W_{it}+u_{it})K_{it}] \Big] }_{B_{T,i}(\theta_i)},
\end{eqnarray*}
so
\[ \sup_{\theta_i \in B_i(\delta)} | S_{T,i}(\theta_{0i}) - \bar{S}_{i}(\theta_{0i}) -  ( S_{T,i}(\theta_i) - \bar{S}_{i}(\theta_i)) | \leq  \sup_{\theta_i \in B_i(\delta)} 2|A_{T,i}(\theta_i)| +\sup_{\theta_i \in B_i(\delta)} |B_{T,i}(\theta_i) -B_{T,i}(\theta_{0i})|. \]
Write $Y_{it} = \theta_{0i}'W_{it} + 0.5 (X_{it}-x)' \ddot{q}_{\tau}(x) (X_{it}-x)^2 + R(x,X_{it})+u_{it}$, so from $\rho_{\tau}(u_1)-\rho_{\tau}(u_2)\leq 2|u_1-u_2|$,
\begin{eqnarray*}
&& \sup_{\theta_i \in B_i(\delta)} |A_{T,i}(\theta_i)| \\
  &\leq&  \frac{C_1}{T h^d}\sum_{t=1}^{T} \left( \|X_{it}-x\|_1^2+\|X_{it}-x\|_1^3 \right)K_{it}  \\
 &= &  \frac{C_1}{T h^d}\sum_{t=1}^{T} \left\{ \left( \|X_{it}-x\|_1^2+\|X_{it}-x\|_1^3 \right)K_{it}  - \mathbb{E}\left[  \left( \|X_{it}-x\|_1^2+\|X_{it}-x\|_1^3 \right)K_{it} \right]  \right\}+\bar{O}(h^2).
 \end{eqnarray*}
Next, consider $\sup_{\theta_i \in B_i(\delta)} |B_{T,i}(\theta_i) -B_{T,i}(\theta_{0i})|$. By compactness of $B_i(\delta)$, there exist a finite $L$ and $ \theta_{i}^{(1)},\ldots, \theta_{i}^{(L)} \in B_i(\delta)$ such that $\| \theta_i^{(j)} - \theta_i^{(k)}\|_1 \leq \omega$ for all $j,k\leq L$. For some $C_2,C_3>0$
\begin{multline*}
\sup_{\theta_i \in B_i(\delta)} |B_{T,i}(\theta_i) -B_{T,i}(\theta_{0i})| \leq \max_{j\leq L} |B_{T,i}(\theta_i^{(j)}) -B_{T,i}(\theta_{0i})| \\
+  \omega C_2 \frac{1}{T h^d}\sum_{t=1}^{T}\Big[ \|W_{it}\|_1 K_{it} -\Ex[\|W_{it}\|_1 K_{it}  ] \Big] + \omega C_3  \Ex[\|W_{it}\|_1 K_{it} /h^d ]. \end{multline*}
Take $\omega = \epsilon/ (6C_3 M)$ where $M = \sup_{i}\Ex[\|W_{it}\|_1 K_{it} /h^d ]<\infty$, then for large enough $N,T $, we have
\begin{multline} \label{A2}
P\Bigg[ \sup_{\theta_i \in B_i(\delta)}| S_{T,i}(\theta_{0i}) - \bar{S}_{i}(\theta_{0i}) -  ( S_{T,i}(\theta_i) - \bar{S}_{i}(\theta_i)) | >\epsilon  \Bigg]  \leq \\
P\Bigg[ \frac{C_1}{T h^d}\sum_{t=1}^{T} \left\{ \left( \|X_{it}-x\|_1^2+\|X_{it}-x\|_1^3 \right)K_{it}  - \mathbb{E}\left[  \left( \|X_{it}-x\|_1^2+\|X_{it}-x\|_1^3 \right)K_{it} \right]  \right\}  > \epsilon/10 C_1\Bigg]
 \\+ \sum_{j=1}^L P\Bigg[|B_{T,i}(\theta_i^{(j)}) -B_{T,i}(\theta_{0i})|  >\epsilon/5 \Bigg]
+ P\Bigg[ \frac{1}{T h^d}\sum_{t=1}^{T}\Big[ \|W_{it}\|_1K_{it} -\Ex[\|W_{it}\|_1K_{it}  ] \Big]  > 6C_3M/5C_2\Bigg].
\end{multline}
Consider the last term on the RHS of \eqref{A2}, since $h  \|W_{it}\|_1K_{it}$ is uniformly bounded, and $\Ex\left[ \|W_{it}\|_1^2 K_{it}^2 \right] = \bar{O}(h^{d})$,  by Bernstein's inequality, we have
\[ P\Bigg[ \frac{1}{T h^d}\sum_{t=1}^{T}\Big[ \|W_{it}\|_1K_{it} -\Ex[\|W_{it}\|_1K_{it}  ] \Big]  > C\Bigg] \leq e^{-\frac{Th^{2d}}{C_4(h^{d} + h^{d-1})}}\]
for some $C_4>0$. We can establish the same bound for the other two terms on the RHS of \eqref{A2}. Thus, \eqref{A1} holds since our assumptions imply that $\log N /(Th^{d+1})\rightarrow 0  $. Thus, it follows that $\|\hat{\phi}-\phi_0\|_1  =o_P(1)$.


Now consider $\hat{\eta}_i$. By definition of the estimators, we have $S_{T,i}(\hat{\eta}_i,\hat{\phi}) \leq S_{T,i}(\eta_{0i},\hat{\phi})$ for all $i\leq N$. Therefore, if $|\hat{\eta}_i -\eta_{0i}|>\delta$ for some $i\leq N$, then we have $\hat{\theta}_i \in B_i^C(\delta)$, and
\[  ( S_{T,i}(\tilde{\theta}_i)  - S_{T,i}(\theta_{0i}) )/r_i  \leq  S_{T,i}(\hat{\theta}_i) -S_{T,i}(\theta_{0i})\leq  S_{T,i}(\eta_{0i},\hat{\phi}) -S_{T,i}(\theta_{0i})  . \]
where $r_i$ and $\tilde{\theta}$ is as define above. Adding the subtracting terms, the above inequality can be written as
\[ \bar{S}_{i}(\tilde{\theta}_i) - \bar{S}_{i}(\theta_{0i}) \leq  S_{T,i}(\theta_{0i}) - \bar{S}_{i}(\theta_{0i}) - [ S_{T,i}(\tilde{\theta}_i) - \bar{S}_{i}(\tilde{\theta}_i)   ] + r_i (S_{T,i}(\eta_{0i},\hat{\phi}) -S_{T,i}(\theta_{0i}) ).\]
Thus, for any $\delta>0$, we have
\begin{multline}\label{A3}
P \left[ \max_{i\leq N}|\hat{\eta}_i -\eta_{0i}| >\delta \right]
\leq \sum_{i=1}^{N} P\left[  \sup_{\theta_i \in B_i(\delta)}| S_{T,i}(\theta_{0i}) - \bar{S}_{i}(\theta_{0i}) -  ( S_{T,i}(\theta_i) - \bar{S}_{i}(\theta_i)) |>C_6 \delta^2 \right] \\
+ P\left[  \max_{1\leq i\leq N} \left| S_{T,i}(\eta_{0i},\hat{\phi}) -S_{T,i}(\theta_{0i}) \right|>C_7 \delta^2 \right]
\end{multline}
for some $C_6,C_7>0$. The first term on the RHS of \eqref{A3} was shown to be $o(1)$ above. To prove $\max_{i\leq N}|\hat{\eta}_i -\eta_{0i}|=o_P(1)$ it suffices to show that for any $\epsilon>0$,
\begin{equation}\label{A4}
 P\Bigg[\max_{i\leq N} |S_{T,i}(\eta_{0i},\hat{\phi}) -S_{T,i}(\theta_{0i} )|>\epsilon \Bigg]=o(1). \end{equation}
Note that
\[ |S_{T,i}(\eta_{0i},\hat{\phi}) -S_{T,i}(\theta_{0i} )| \lesssim \|\hat{\phi} -\phi_0\|_1  \cdot \frac{1}{T h^d}\sum_{t=1}^{T}\| (X_{it}-x)/h\|_1 K_{it}. \]
Moreover,
\begin{multline}\label{A5}  \max_{1\leq i\leq N }\frac{1}{T h^d}\sum_{t=1}^{T}\| (X_{it}-x)/h\|_1 K_{it} \leq \max_{1\leq i\leq N }\frac{1}{T h^d}\sum_{t=1}^{T}  \left\{ \| (X_{it}-x)/h\|_1 K_{it}  -\Ex[| (X_{it}-x)/h\|_1 K_{it} ]  \right\}\\
+\max_{1\leq i\leq N } \Ex[| (X_{it}-x)/h\|_1 K_{it}/h^d ].
\end{multline}
Similar to the proof above, we can show that the first term on the RHS of \eqref{A5} is $o_P(1)$, and the second term is $\max_{1\leq i \leq N}f_{X,i}(x) +o(1)<\infty$ by our assumptions.
Then \eqref{A4} follows by $  \|\hat{\phi} -\phi_0\|_1=o_P(1)$. This concludes the proof.
\end{proof}

\vspace{0.5cm}

Define
\[ \bar{S}_i^{\eta_i}(\theta_i) =\Ex\Bigg[     [\mathbf{1}\{ u_{it} \leq (\theta_i-\theta_{0i})'W_{it}\}-\tau] \frac{K_{it}}{h^d}       \Bigg], \quad \bar{S}_i^{\phi}(\theta_i) =\Ex\Bigg[     [\mathbf{1}\{ u_{it} \leq (\theta_i-\theta_{0i})'W_{it}\}-\tau]  \frac{X_{it}-x}{h} \frac{K_{it}}{h^d}       \Bigg].\]
Denote $\bar{S}_i^{\eta_i\eta_i}(\theta_i)=\partial \bar{S}_i^{\eta_i}(\theta_i)/\partial \eta_i$, and $\bar{S}_i^{\phi\phi}(\theta_i), \bar{S}_i^{\phi\eta_i}(\theta_i)$ etc. are defined in a similar fashion. The arguments of these functions are dropped when they are evaluated at $\theta_{0i}$. Let $f_i(x) = f_{u,i}(0|x)f_{X,i}(x)$, $\bar{f}_N(x)=N^{-1}\sum_{i=1}^{N}f_i(x)$, $\bar{f}(x)=\lim_{N\rightarrow\infty}\bar{f}_N(x)$.

\vspace{0.5cm}
\begin{lemma}
For a boundary point $x$, i.e., $x=ch$ for some $c$ in the support of $K(\cdot)$, we have
\[ \bar{S}_i^{\phi\phi} =\mathcal{C}_2  f_i(0)+\bar{O}(h), \quad \bar{S}_i^{\phi\eta_i} = \mathcal{C}_1 f_i(0)+\bar{O}(h) \quad  \bar{S}_i^{\eta_i\eta_i} = c_0 f_i(0)+\bar{O}(h). \]
For a interior point $x$, we have
\[ \bar{S}_i^{\phi\phi} = \mathcal{K}_1 f_i(x)+\bar{O}(h), \quad \bar{S}_i^{\phi\eta_i} = \bar{O}(h) \quad  \bar{S}_i^{\eta_i\eta_i} =f_i(x)+\bar{O}(h). \]
\end{lemma}
\begin{proof}
The proof follows from standard calculations for kernel density estimators. Therefore, it is omitted.
\end{proof}

\vspace{0.5cm}
\begin{lemma} The following representation holds under Assumptions (A1) to (A6):

For a boundary point $x$, i.e., $x=ch$ for some $c>0$ in the support of $K(\cdot)$, we have
\[   \hat{\phi}-\phi_0 -h^2 B^{(1)}=-  \frac{\mathcal{C}^{-1}}{\bar{f}_N(0)}\frac{1}{NT}\sum_{i=1}^{N}\sum_{t=1}^{T}[ \mathbf{1}\{ u_{it}\leq  0\} -\tau ] \Bigg[ \frac{X_{it}-x}{h}-\frac{\mathcal{C}_1}{c_0}\Bigg]     \frac{K_{it}}{h^d}  +o_P(1/\sqrt{NTh^d}). \]
For a interior point $x$, we have
\[ \hat{\phi}-\phi_0 = -  \frac{\mathcal{K}_1^{-1}}{\bar{f}_N(x)}\frac{1}{NT}\sum_{i=1}^{N}\sum_{t=1}^{T}[ \mathbf{1}\{ u_{it}\leq  0\} -\tau ] \frac{X_{it}-x}{h}    \frac{K_{it}}{h^d} +o_P(1/\sqrt{NTh^d}). \]

\end{lemma}

\begin{proof}
We only provide the proof for the boundary point, which is more involved.
Let $\{\delta_{NT}\}$ be a non-increasing sequence such that $\max_{i\leq N}|\hat{\eta}_i-\eta_{0i}|=O_P(\delta_{NT})$, and let $\{\gamma_{NT}\}$ be a non-increasing sequence such that
$ \max_{i\leq N}|\hat{\eta}_i -\eta_{0i}| \vee \|\hat{\phi}-\phi_0\| =O_P(\gamma_{NT})$.

\noindent \textbf{Step 1 (Expansion): }

Expanding $\bar{S}_i^{\eta_i}(\hat{\theta}_i)$ and $\bar{S}_i^{\phi}(\hat{\theta}_i)$ around $\theta_{0i}$ gives:
\begin{eqnarray*}
\bar{S}_i^{\phi}(\hat{\theta}_i) &=& \bar{S}_i^{\phi\phi} (\hat{\phi}-\phi_0) +  \bar{S}_i^{\phi\eta_i} (\hat{\eta}_i-\eta_{0i})+\bar{O}_P(\gamma_{NT}) \cdot \|\hat{\phi}-\phi_0\| + \Bar{O}_P( \delta_{NT}^2 ) ,    \\
\bar{S}_i^{\eta_i}(\hat{\theta}_i) &=& ( \bar{S}_i^{\eta_i\phi})' (\hat{\phi}-\phi_0) +  \bar{S}_i^{\eta_i\eta_i} (\hat{\eta}_i-\eta_{0i}) +\bar{O}_P(\gamma_{NT})\cdot \|\hat{\phi}-\phi_0\| + \Bar{O}_P( \delta_{NT}^2 ).
\end{eqnarray*}
Plugging the second equation into the first one, and using the results of Lemma 1 and Lemma 2 gives
\begin{equation*}
[f_i(0)(\mathcal{C}_2-\mathcal{C}_1 \mathcal{C}_1'/c_0)+\bar{O}(h)] \cdot (\hat{\phi}-\phi_0) = \bar{S}_i^{\phi}(\hat{\theta}_i) - (\bar{S}_i^{\phi\eta_i}/\bar{S}_i^{\eta_i\eta_i}) \cdot \bar{S}_i^{\eta_i}(\hat{\theta}_i) +\bar{O}_P(\gamma_{NT})\|\hat{\phi}-\phi_0\| + \Bar{O}_P( \delta_{NT}^2 ).
\end{equation*}
It follows that
\begin{multline}\label{A6}
[  \bar{f}_N(0)(\mathcal{C}_2-\mathcal{C}_1 \mathcal{C}_1'/c_0)] \cdot (\hat{\phi}-\phi_0) = \frac{1}{N}\sum_{i=1}^N\bar{S}_i^{\phi}(\hat{\theta}_i) - \frac{1}{N}\sum_{i=1}^N(\bar{S}_i^{\phi\eta_i}/\bar{S}_i^{\eta_i\eta_i})\cdot \bar{S}_i^{\eta_i}(\hat{\theta}_i) \\+ O_P(h+\gamma_{NT})\|\hat{\phi}-\phi_0\| + O_P( \delta_{NT}^2 ).
\end{multline}
Similarly, for each $i$ we can obtain:
\begin{equation}\label{A7}
[f_i(0)c_0+\bar{O}(h)] \cdot (\hat{\eta}_i-\eta_{0i}) =  \bar{S}_i^{\eta_i}(\hat{\theta}_i) -  f_i(0)\mathcal{C}_1' (\hat{\phi}-\phi_0) +\bar{O}_P(h+\gamma_{NT})\|\hat{\phi}-\phi_0\| + \Bar{O}_P( \delta_{NT}^2 ).
\end{equation}

\noindent \textbf{Step 2 ($N^{-1}\sum_{i=1}^N\bar{S}_i^{\phi}(\hat{\theta}_i)$ and $N^{-1}\sum_{i=1}^N\bar{S}_i^{\eta_i}(\hat{\theta}_i)$):}

First consider $N^{-1}\sum_{i=1}^N\bar{S}_i^{\phi}(\hat{\theta}_i)$. Define $Z_{it}= (X_{it}-x)/h\cdot K_{it}/h^d$. We can write
\begin{multline*}
\frac{1}{N}\sum_{i=1}^N\bar{S}_i^{\phi}(\hat{\theta}_i) = \frac{1}{NT}\sum_{i=1}^{N}\sum_{t=1}^{T} [ \mathbf{1}\{ Y_{it}\leq \hat{\theta}_i ' W_{it}\} -\tau ]Z_{it}  \\
 - \frac{1}{NT}\sum_{i=1}^{N}\sum_{t=1}^{T} \Big\{[ \mathbf{1}\{ u_{it}\leq  (\hat{\theta}_i-\theta_{0i}) ' W_{it}\} -\tau ]Z_{it}
-\Ex\Big[[ \mathbf{1}\{ u_{it}\leq  ( \hat{\theta}_i-\theta_{0i}) ' W_{it}\} -\tau ]Z_{it}   \Big]\Big\} \\
+ \frac{1}{NT}\sum_{i=1}^{N}\sum_{t=1}^{T} \big[    \mathbf{1}\{ u_{it}\leq  (\hat{\theta}_i-\theta_{0i}) ' W_{it}\} -  \mathbf{1}\{ Y_{it}\leq \hat{\theta}_i ' W_{it}\}  \big]Z_{it}.
\end{multline*}
Define
\[
A_{NT}(\theta) =  \frac{1}{NT}\sum_{i=1}^{N}\sum_{t=1}^{T} \Big\{[ \mathbf{1}\{ u_{it}\leq  (\theta_i-\theta_{0i}) ' W_{it}\} -\tau ]Z_{it}
-\Ex\Big[[ \mathbf{1}\{ u_{it}\leq  (\theta_i-\theta_{0i}) ' W_{it}\} -\tau ]Z_{it}   \Big]\Big\}
\]
and
\[
B_{NT}(\theta)=\frac{1}{NT}\sum_{i=1}^{N}\sum_{t=1}^{T} \big[    \mathbf{1}\{ u_{it}\leq  (\theta_i-\theta_{0i}) ' W_{it}\} -  \mathbf{1}\{ Y_{it}\leq \theta_i ' W_{it}\}  \big]Z_{it},
\]
then adding and subtracting terms, we can write
\begin{multline}\label{A8}
\frac{1}{N}\sum_{i=1}^N\bar{S}_i^{\phi}(\hat{\theta}_i) = \frac{1}{NT}\sum_{i=1}^{N}\sum_{t=1}^{T} [ \mathbf{1}\{ Y_{it}\leq \hat{\theta}_i ' W_{it}\} -\tau ]Z_{it}
- A_{NT}(\theta_0) +( A_{NT}(\theta_0)  - A_{NT}(\hat{\theta}) ) +B_{NT}(\hat{\theta}).
\end{multline}
Moreover, define
\begin{multline*} B_{NT}^{(1)}(\theta)= \frac{1}{NT}\sum_{i=1}^{N}\sum_{t=1}^{T} \Big\{  \big[    \mathbf{1}\{ u_{it}\leq  (\theta_i-\theta_{0i}) ' W_{it}\} -  \mathbf{1}\{ Y_{it}\leq \theta_i ' W_{it}\}  \big]Z_{it}   \\-\Ex\Big[  \big[    \mathbf{1}\{ u_{it}\leq  (\theta_i-\theta_{0i}) ' W_{it}\} -  \mathbf{1}\{ Y_{it}\leq \theta_i ' W_{it}\}  \big]Z_{it} \Big] \Big\}
\end{multline*}
\[ B_N^{(2)}(\theta) = \frac{1}{N}\sum_{i=1}^{N}\Ex\Big[  \big[    \mathbf{1}\{ u_{it}\leq  (\theta_i-\theta_{0i}) ' W_{it}\} -  \mathbf{1}\{ Y_{it}\leq \theta_i ' W_{it}\}  \big]Z_{it}, \Big]\]
%
we can write:
\begin{equation*}B_{NT}(\hat{\theta}) = B_{NT}^{(1)}(\theta_0) + ( B_{NT}^{(1)}(\hat{\theta}) -B_{NT}^{(1)}(\theta_0))  +B_{N}^{(2)}(\hat{\theta}). \end{equation*}

First, by the computational property of quantile regressions, the first term on the right-hand side of \eqref{A8} is $O_P(T^{-1}h^{-d-1})$.

Second, consider $ A_{NT}(\theta_0)  - A_{NT}(\hat{\theta})$. Following the proof of \citet{kato2012asymptotics}, we will show that
\begin{equation}\label{A9}
A_{NT}(\theta_0)  - A_{NT}(\hat{\theta}) =O_P(\xi_{NT}) \quad \text{where } \quad \xi_{NT} = \frac{|\log (h^{d+2}\gamma_{NT} ) |}{Th^{d+1}} \vee  \frac{\sqrt{\gamma_{NT}}\cdot \sqrt{|\log (h^{d+2}\gamma_{NT} ) |}}{\sqrt{Th^d}}.
\end{equation}
Note that
\[ | A_{NT}(\theta_0)  - A_{NT}(\hat{\theta})| \leq \frac{1}{Nh^{d+1}}\sum_{i=1}^{N}  \sup_{ |\eta_i-\eta_{0i}|, |\phi-\phi_0|\leq \gamma_{NT}}\Bigg|\frac{1}{T}\sum_{t=1}^{T} [l(\eta_i,\phi,V_{it}) - \Ex l(\eta_i,\phi,V_{it})  ]  \Bigg| ,\]
where $V_{it}=[Y_{it},X_{it}]$, and $l(\eta_i,\phi,V_{it}) =  [ \mathbf{1}\{ u_{it}\leq  (\theta_i-\theta_{0i}) ' W_{it}\} - \mathbf{1}\{ u_{it}\leq 0\} ]Z_{it}h^{d+1} $. Define $\mathcal{L}_{\gamma} = \{l(\eta_i,\phi,V_{it}): | \eta_i-\eta_{0i}| \vee \|\phi-\phi_0\|\leq \gamma\}$, then the above inequality can be written as
\[ | A_{NT}(\theta_0)  - A_{NT}(\hat{\theta})| \leq \frac{1}{Nh^{d+1}}\sum_{i=1}^{N} \Bigg\| \frac{1}{T}\sum_{t=1}^{T}  l(V_{it}) -\Ex l(V_{it})\Bigg\|_{\mathcal{L}_{\gamma_{NT}}}.\]
Thus, to prove \eqref{A9}, it suffices to show that
\begin{equation}\label{A10}
\max_{i\leq N} \Ex \Bigg\| \frac{1}{T}\sum_{t=1}^{T}  l(V_{it}) -\Ex l(V_{it})\Bigg\|_{\mathcal{L}_{\gamma_{NT}}} =O_P \Bigg(\frac{|\log (h^{d+2}\gamma_{NT} ) |}{T} \vee  \frac{\sqrt{h^{d+2}\gamma_{NT}}\cdot \sqrt{|\log (h^{d+2}\gamma_{NT} ) |}}{\sqrt{T}}\Bigg).
\end{equation}
Since $Z_{it}h^{d+1} =K_{it}\cdot(X_{it}-x)$ is uniformly bounded, the class of functions $\mathcal{L}_{\infty} = \{l(\eta_i,\phi,V_{it}): \phi,\eta_i\in\mathbb{R}\}$ is a VC subgraph class, and $\Ex[l(\eta_i,\phi,V_{it})^2 ]
=O(\gamma_{NT} h^{d+2})$ for $l(\eta_i,\phi,V_{it}) \in \mathcal{L}_{\gamma_{NT}} $, \eqref{A10} follows from Proposition B.1 of \citet{kato2012asymptotics}. Similarly, we can show that
\begin{equation}\label{A11}
 B_{NT}^{(1)}(\hat{\theta}) -B_{NT}^{(1)}(\theta_0) = O_P(\xi_{NT}).
\end{equation}

Third, consider $B_{NT}^{(1)}(\theta_0)$. Note that
\[B_{NT}^{(1)}(\theta_0) =  \frac{1}{NT}\sum_{i=1}^{N}\sum_{t=1}^{T} \Big\{  \big[    \mathbf{1}\{ u_{it}\leq  0\} -  \mathbf{1}\{ u_{it}\leq  -m(X_{it})\}  \big]Z_{it}   -\Ex\Big[  \big[    \mathbf{1}\{ u_{it}\leq  0\} -  \mathbf{1}\{ u_{it}\leq  -m(X_{it})\}  \big]Z_{it} \Big] \Big\} \]
where $Y_{it}=\theta_0'W_{it}+0.5 (X_{it}-x)' \ddot{q}_{\tau}(x)(X_{it}-x)  +R(X_{it},x)+u_{it}$ and $m(X_{it}) =  0.5 (X_{it}-x)' \ddot{q}_{\tau}(x)(X_{it}-x) + R(X_{it},x)$. Since
\[ \Ex\Big\| B_{NT}^{(1)}(\theta_0)\Big\|^2
\leq \frac{1}{N^2 T} \sum_{i=1}^{N} \Ex \Big[ | \mathbf{1}\{ u_{it}\leq  0\} -  \mathbf{1}\{ u_{it}\leq  -m(X_{it})\} | \cdot\|Z_{it}\|^2 \Big] =O(h^{2-d}/(NT)) ,
\]
it follows that
\begin{equation}\label{A12} B_{NT}^{(1)}(\theta_0) =o_P(1/\sqrt{NTh^d}). \end{equation}

Finally, consider $B_{N}^{(2)}(\hat{\theta})$. Write
\[ B_{N}^{(2)}(\theta) = \frac{1}{N}\sum_{i=1}^{N}\Ex\Big[  \big[    \mathbf{1}\{ u_{it}\leq  (\theta_i-\theta_{0i}) ' W_{it}\} -  \mathbf{1}\{  u_{it}\leq  (\theta_i-\theta_{0i}) ' W_{it} -m(X_{it})\}  \big]Z_{it} \Big],\]
since for some $c\in(0,1)$ we have
\begin{eqnarray*}
&& \Ex\Big[  \big[    \mathbf{1}\{ u_{it}\leq  (\theta_i-\theta_{0i}) ' W_{it}\} -  \mathbf{1}\{  u_{it}\leq  (\theta_i-\theta_{0i}) ' W_{it} -m(X_{it})\}  \big]Z_{it} \Big]  \\
&=& \Ex\Big[  f_{u,i}( (\theta_i-\theta_{0i}) ' W_{it} -c\cdot m(X_{it})  |X_{it})m(X_{it})               Z_{it} \Big] \\
&= &\Ex\Big[  f_{u,i}( 0  |X_{it})m(X_{it})               Z_{it} \Big]+ \Ex\Big[  [f_{u,i}( (\theta_i-\theta_{0i}) ' W_{it} -c\cdot m(X_{it})  |X_{it})- f_{u,i}( 0 |X_{it}) ]m(X_{it})               Z_{it} \Big]
\end{eqnarray*}
where the first term of the last expression is $0.5f_i(0)  h^2\int_{\mathcal{B}} u' \ddot{q}_{\tau}(0)u uK(u)du  +\bar{O}(h^3)$, and the second term of the last expression is bounded by
\[ C \|  \theta_i-\theta_{0i} \| \cdot \Ex[ \|  W_{it}\| \cdot |m(X_{it})|\cdot \|Z_{it}\|] +C \Ex[  | m(X_{it})|^2 \|Z_{it}\|] =\|  \theta_i-\theta_{0i} \|\cdot \bar{O}(h^2) +\bar{O}(h^4).\]
It then follows that
\begin{equation}\label{A13}
B_{N}^{(2)}(\hat{\theta}) =0.5  h^2  \bar{f}_N(0)\int_{\mathcal{B}} u' \ddot{q}_{\tau}(0)u uK(u)du +O_P(\gamma_{NT}\cdot h^2) +O_P(h^3). \end{equation}

Combining results \eqref{A8}, \eqref{A9}, \eqref{A11}, \eqref{A12} and \eqref{A13} gives
\begin{multline}\label{A14}
\frac{1}{N}\sum_{i=1}^N\bar{S}_i^{\phi}(\hat{\theta}_i) = - \frac{1}{NT}\sum_{i=1}^{N}\sum_{t=1}^{T} [ \mathbf{1}\{ u_{it}\leq  0\} -\tau ]Z_{it}  +0.5  h^2  \bar{f}_N(0)\int_{\mathcal{B}} u' \ddot{q}_{\tau}(0)u uK(u)du   \\+O_P(1/ (Th^{d+1})) +O_P(\xi_{NT})+o_P(1/\sqrt{NTh^d}) +O_P(\gamma_{NT}\cdot h^2)+O_P(h^3).
\end{multline}
Similarly, we can show that
\begin{multline}\label{A15}
\frac{1}{N}\sum_{i=1}^N(\bar{S}_i^{\phi\eta_i}/\bar{S}_i^{\eta_i\eta_i})  \cdot \bar{S}_i^{\eta_i}(\hat{\theta}_i) = - \frac{1}{NT}\sum_{i=1}^{N}\sum_{t=1}^{T} (\bar{S}_i^{\phi\eta_i}/\bar{S}_i^{\eta_i\eta_i})  \cdot [ \mathbf{1}\{ u_{it}\leq  0\} -\tau ]K_{it}/h^d \\ +0.5  h^2 \bar{f}_N(0) \mathcal{C}_1/c_0 \cdot \int_{\mathcal{B}} u' \ddot{q}_{\tau}(0)u K(u)du  +O_P(1/ (Th^{d+1})) \\+O_P(\xi_{NT})+o_P(1/\sqrt{NTh^d}) +O_P(\gamma_{NT}h^2)+O_P(h^3).
\end{multline}

\noindent \textbf{Step 3 (Rate of convergence):}

Plugging \eqref{A14} and \eqref{A15} into \eqref{A6} and using $\bar{S}_i^{\phi\eta_i}/\bar{S}_i^{\eta_i\eta_i} = \mathcal{C}_1/c_0+\bar{O}(h)$ give
\begin{multline}\label{A16}
 [ \bar{f}_N(0)(\mathcal{C}_2-\mathcal{C}_1 \mathcal{C}_1'/c_0)] \cdot (\hat{\phi}-\phi_0) \\=- \frac{1}{NT}\sum_{i=1}^{N}\sum_{t=1}^{T} [ \mathbf{1}\{ u_{it}\leq  0\} -\tau ] \Bigg[ \frac{X_{it}-x}{h}-\frac{\mathcal{C}_1}{c_0}\Bigg]     \frac{K_{it}}{h^d} + 0.5h^2 \bar{f}_N(0)\int_{\mathcal{B}} u' \ddot{q}_{\tau}(0)u \Big(u-\frac{\mathcal{C}_1}{c_0}\Big)K(u)du   \\    +o_P\left(\|\hat{\phi}-\phi_0\| \right)  +O_P( \delta_{NT}^2 )+O_P(1/ (Th^{d+1})) +O_P(\xi_{NT})+o_P(1/\sqrt{NTh^d}) +O_P(h^3)+O_P(\gamma_{NT}h^2).
\end{multline}
It then follows from our assumptions that
\begin{equation}\label{A17}
\hat{\phi}-\phi_0 = o_P(1/\sqrt{Th^d})+O_P( \delta_{NT}^2 )+O_P(h^2).
\end{equation}

Next, the above inequality and \eqref{A7} imply that for some $C>0$,
\[ \max_{i\leq N}|\hat{\eta}_i - \eta_{0i}|  \leq C \max_{i\leq N}|  \bar{S}_i^{\eta_i}(\hat{\theta}_i)| +o_P(1/\sqrt{Th^d}) +O_P(h^2).\]
Similar to the proof of Step 2, we can show that
\[  \bar{S}_i^{\eta_i}(\hat{\theta}_i) =\bar{O}_P(T^{-1}h^{-d-1}) +\bar{O}_P(h^2) +\bar{o}_P(1/\sqrt{Th^d})- D_{T}(\theta_{0i}) +( D_{T}(\theta_{0i})  - D_{T}(\hat{\theta}_i) )  +G_{T}(\hat{\theta}_i) -G_{T}(\theta_{0i})\]
where
\[
D_{T}(\theta_i) =  \frac{1}{T}\sum_{t=1}^{T} \Big\{[ \mathbf{1}\{ u_{it}\leq  (\theta_i-\theta_{0i}) ' W_{it}\} -\tau ]K_{it}/h^d
-\Ex\Big[[ \mathbf{1}\{ u_{it}\leq  (\theta_i-\theta_{0i}) ' W_{it}\} -\tau ]K_{it}/h^d   \Big]\Big\},
\]
\begin{multline*}
G_{T}(\theta_i) = \frac{1}{T}\sum_{t=1}^{T} \Big\{  \big[    \mathbf{1}\{ u_{it}\leq  (\theta_i-\theta_{0i}) ' W_{it}\} -  \mathbf{1}\{ Y_{it}\leq \theta_i ' W_{it}\}  \big]K_{it}/h^d   \\-\Ex\Big[  \big[    \mathbf{1}\{ u_{it}\leq  (\theta_i-\theta_{0i}) ' W_{it}\} -  \mathbf{1}\{ Y_{it}\leq \theta_i ' W_{it}\}  \big]K_{it}/h^d \Big] \Big\}.
\end{multline*}
Applying Lemma 2.2.9 and Lemma 2.2.10 of \citet{vanweak} we can show that
\[
\max_{i\leq N}|  D_{T}(\theta_{0i})| = O_P(\sqrt{\log N}/\sqrt{Th^d}).
 \]
Next, following the proof of \citet{kato2012asymptotics}, we can show that:
\[ \max_{i\leq N}| D_{T}(\hat{\theta}_{i}) - D_{T}(\theta_{0i})| =o_P(\sqrt{\log N}/\sqrt{Th^d}), \text{ and }\max_{i\leq N}| G_{T}(\hat{\theta}_{i}) - G_{T}(\theta_{0i})| =o_P(\sqrt{\log N}/\sqrt{Th^d}) \]
It then follows that
\[\max_{i\leq N } |\bar{S}_i^{\eta_i}(\hat{\theta}_i)|=O_P(\sqrt{\log N}/\sqrt{Th^d})  +O_P(h^2), \]
\begin{equation}\label{A18}
 \max_{i\leq N}|\hat{\eta}_i - \eta_{0i}|  = O_P(\sqrt{\log N}/\sqrt{Th^d})  +O_P(h^2),
\end{equation}
\begin{equation}\label{A19}
 \|\hat{\phi} - \phi_{0}\|  = o_P(\sqrt{\log N}/\sqrt{Th^d})  +O_P(h^2).
\end{equation}

\noindent \textbf{Step 4 (Bahadur Representation):}

By \eqref{A18}, \eqref{A19} we have $\delta_{NT} = \gamma_{NT} = \sqrt{\log N}/\sqrt{Th^d}$. It then follows from \eqref{A16} that
\begin{multline}
  \hat{\phi}-\phi_0 - h^2  B^{(1)}
  =-  \frac{\mathcal{C}^{-1}}{\bar{f}_N(0)}\frac{1}{NT}\sum_{i=1}^{N}\sum_{t=1}^{T}[ \mathbf{1}\{ u_{it}\leq  0\} -\tau ] \Bigg[ \frac{X_{it}-x}{h}-\frac{\mathcal{C}_1}{c_0}\Bigg]     \frac{K_{it}}{h^d} + o_P\left(\|\hat{\phi}-\phi_0\| \right) \\+ O_P(1/(Th^{d+1})) + O_P(\log N/Th^d ) +O_P(\xi_{NT})+ o_P(1/\sqrt{NTh^d}) +O_P(h^3)+O_P(\gamma_{NT}h^2).
\end{multline}
Assumption (A6) implies that $ O_P(1/(Th^{d+1}))$, $O_P(\log N/Th^d )$, $O_P(\gamma_{NT}h^2)$ and $O_P(h^3)$ are all $o_P(\sqrt{NTh^d})$. Further, it can be shown that
\[ O_P(\xi_{NT}) =O_P\Bigg(\frac{(\log N)^{1/4}(\log T)^{1/2}}{(Th^d)^{3/4}}\Bigg),\]
which is also $o_P(\sqrt{NTh^d})$ by Assumption (A6). So the desired result for the boundary point follows. The desired result for the interior points can be proved in the same way, by noting that for interior points $c_0=1,\mathcal{C}_1=0$, $\mathcal{C}=\mathcal{K}_1$ and $b=0$. \end{proof}

\noindent{\textbf{Proof of Theorem 1:}}

\begin{proof}
By Lyapunov's CLT, we can show that:
\begin{equation}\label{A21}
\frac{1}{\sqrt{NT}}\sum_{i=1}^{N}\sum_{t=1}^{T}[ \mathbf{1}\{ u_{it}\leq  0\} -\tau ] \Bigg[ \frac{X_{it}-x}{h}-\frac{\mathcal{C}_1}{c_0}\Bigg]     \frac{K_{it}}{ h^{d/2}}
\overset{d}{\rightarrow} \mathcal{N}\left(0, \tau(1-\tau)\Omega^{\ast}  \right),
\end{equation}
where
\[ \Omega^{\ast}= \left( \lim_{N\rightarrow\infty} N^{-1}\sum_{i=1}^{N}f_{X,i}(0)\right) \cdot\int_{\mathcal{B}}(u-\mathcal{C}_1/c_0)(u-\mathcal{C}_1/c_0)' K^2(u)du  \]
then the desired result for the boundary points follows from \eqref{A21} and Lemma 3. The desired result for the interior points follows similarly. \end{proof}

\subsection{Proof of Theorem 2}
We write $\check{\eta}_i,\check{\beta}$ instead of $\check{\eta}_{i,\tau}(x), \check{\beta}_{\tau}(x)$. Moreover, let $ \check{\phi}=h\check{\beta}$, $\check{\theta}_i = (\check{\eta}_i, \check{\phi}')'$, $\check{\theta} = (\check{\eta}_1,\ldots, \check{\eta}_N, \check{\phi}')'$, and define $\varrho_{\tau}(u) = [\tau - G(u/b)] u$,
\[ S_{T,i}^{\ast}(\theta_i) = \frac{1}{T h^d}\sum_{t=1}^{T}\varrho_{\tau}(Y_{it} - \theta_i'W_{it})K_{it}, \quad  \bar{S}_{i}^{\ast}(\theta_i) = \Ex[ \varrho_{\tau}( (\theta_{0i} - \theta_i)'W_{it}+u_{it})K_{it}/h^d],\]
and $S_{NT}^{\ast}(\theta)  = N^{-1}\sum_{i=1}^{N}S_{T,i}^{\ast}(\theta_i) $, $\bS_{N}^{\ast}(\theta) = N^{-1}\sum_{i=1}^{N} \bar{S}_{i}^{\ast}(\theta_i)$.

We only prove the result for the boundary points. The proof for the interior points is almost the same.

\vspace{0.5cm}

\begin{lemma} Under Assumptions B1 to B4, we have $ \|\check{\phi}-\phi_0\|_1 =o_P(1)$ and $\max_{i\leq N}|\check{\eta}_i -\eta_{0i}|=o_P(1)$.
\end{lemma}

\begin{proof}
First consider $\check{\phi}$. By the definition of the estimators, there exists some $i\leq N$ such that $S_{T,i}^{\ast}(\check{\theta}_i) \leq S_{T,i}^{\ast}(\theta_{0i}) $.
For any $\delta>0$, $B_i(\delta)$ is as defined in the proof of Lemma 1. Suppose that $\|\check{\phi}-\phi_0\|_1>\delta$, then $\check{\theta}_i \in B^C_i(\delta) $. Similarly, by the convexity of $S_{T,i}$, we have
\[ S_{T,i}(\tilde{\theta}_i)  - S_{T,i}(\theta_{0i})  \leq \check{r}_i  (S_{T,i}(\check{\theta}_i) -S_{T,i}(\theta_{0i}) ) ,   \]
where $\tilde{\theta}_i  = \check{r}_i \check{\theta}_i+(1-\check{r}_i) \theta_{0i} $  is on the boundary of $B_i(\delta)$ and $ \check{r}_i= \delta/(|\check{\eta}_i-\eta_{0i} |+\|\check{\phi}-\phi_0\|_1)<1$. Adding and subtracting terms gives
\[ \bar{S}_{i}(\tilde{\theta}_i) - \bar{S}_{i}(\theta_{0i}) \leq  S_{T,i}(\theta_{0i}) - \bar{S}_{i}(\theta_{0i}) - [ S_{T,i}(\tilde{\theta}_i) - \bar{S}_{i}(\tilde{\theta}_i)   ]  +\check{r}_i  (S_{T,i}(\check{\theta}_i) -S_{T,i}(\theta_{0i}) -S_{T,i}^{\ast}(\check{\theta}_i) +S_{T,i}^{\ast}(\theta_{0i})) . \]
The last term on the right-hand side of the above inequality is $\bar{O}_P(b)$, because we can show that $\sup_{\theta_i} | S_{T,i}(\theta_i) -S_{T,i}^{\ast}(\theta_i) | \leq  Cb \cdot T^{-1}\sum_{t=1}^{T} K_{it}/h^d$ for some $C<\infty$ (see \citealt{horowitz1998bootstrap}), and it is easy to show that $\sup_{i\leq N}T^{-1}\sum_{t=1}^{T} K_{it}/h^d = O_P(1)+O_P(\sqrt{\log N/Th^d})$. The rest of the proof is similar to the proof of Lemma 1 given that $b\rightarrow0$ as $N,T\rightarrow\infty$.

Now consider $\check{\eta}_i$. By definition of the estimators, we have $S_{T,i}^{\ast}(\check{\eta}_i,\check{\phi}) \leq S_{T,i}^{\ast}(\eta_{0i},\check{\phi})$ for all $i\leq N$. Therefore, if $|\check{\eta}_i -\eta_{0i}|>\delta$ for some $i\leq N$, then we have $\check{\theta}_i \in B_i^C(\delta)$, and
\[  ( S_{T,i}(\tilde{\theta}_i)  - S_{T,i}(\theta_{0i}) )/\check{r}_i  \leq  S_{T,i}(\check{\theta}_i) -S_{T,i}(\theta_{0i}) . \]
where $r_i$ and $\tilde{\theta}$ is as define above. Adding the subtracting terms, we can write
\[ S_{T,i}(\check{\theta}_i) -S_{T,i}(\theta_{0i}) =S_{T,i}^{\ast}(\check{\theta}_i) - S_{T,i}^{\ast}(\eta_{0i},\check{\phi})+S_{T,i}^{\ast}(\eta_{0i},\check{\phi})   -S_{T,i}^{\ast}(\theta_{0i})+  \left[ S_{T,i}(\check{\theta}_i) -S_{T,i}^{\ast}(\check{\theta}_i)  -S_{T,i}(\theta_{0i}) +S_{T,i}^{\ast}(\theta_{0i})\right]. \]
Thus, from $S_{T,i}^{\ast}(\check{\eta}_i,\check{\phi}) \leq S_{T,i}^{\ast}(\eta_{0i},\check{\phi})$ we have
\begin{multline*}
\bar{S}_{i}(\tilde{\theta}_i) - \bar{S}_{i}(\theta_{0i}) \leq  S_{T,i}(\theta_{0i}) - \bar{S}_{i}(\theta_{0i}) - [ S_{T,i}(\tilde{\theta}_i) - \bar{S}_{i}(\tilde{\theta}_i)   ] \\
+ \check{r}_i \left(S_{T,i}^{\ast}(\eta_{0i},\check{\phi})   -S_{T,i}^{\ast}(\theta_{0i})+  \left[ S_{T,i}(\check{\theta}_i) -S_{T,i}^{\ast}(\check{\theta}_i)  -S_{T,i}(\theta_{0i}) +S_{T,i}^{\ast}(\theta_{0i})\right]  \right).
\end{multline*}
Since $\| \check{\phi} - \phi_0\|_1=o_P(1)$, the last term on the RHS of the above inequality is $\bar{o}_P(1)+\bar{O}_P(b)$. The rest of the proof is similar to the proof of Lemma 1. \end{proof}


%
%

\vspace{0.5cm}

\begin{lemma} Let $\mathcal{B}_{i,\delta}$ be a neighbourhood of $\theta_{0i}$, then under Assumptions B1 to B4, we have
\[
\sup_{\theta_i\in \mathcal{B}_{i,\delta} }  \| S_{T,i}^{\ast \phi \phi \phi_j}(\theta_i) \| =\bar{O}_P(1),  \sup_{\theta_i\in \mathcal{B}_{i,\delta} }  \| S_{T,i}^{\ast \eta_i \phi \phi_j}(\theta_i) \| =\bar{O}_P(1),   \sup_{\theta_i\in \mathcal{B}_{i,\delta} }  \| S_{T,i}^{\ast \phi \phi \eta_i}(\theta_i) \| =\bar{O}_P(1),
\]
\[ \sup_{\theta_i\in \mathcal{B}_{i,\delta} }  \| S_{T,i}^{\ast \eta_i \phi \eta_i}(\theta_i) \| =\bar{O}_P(1), \sup_{\theta_i\in \mathcal{B}_{i,\delta} }  \| S_{T,i}^{\ast \phi \eta_i \eta_i}(\theta_i) \| =\bar{O}_P(1),  \sup_{\theta_i\in \mathcal{B}_{i,\delta} }  | S_{T,i}^{\ast \eta_i \eta_i \eta_i}(\theta_i) | =\bar{O}_P(1),  \]
\[
 \sup_{\theta_i\in \mathcal{B}_{i,\delta} }  \| S_{T,i}^{\ast \phi \eta_i \phi_j}(\theta_i) \| =\bar{O}_P(1),  \sup_{\theta_i\in \mathcal{B}_{i,\delta} }  | S_{T,i}^{\ast \eta_i \eta_i \phi_j}(\theta_i) | =\bar{O}_P(1).
\]
\end{lemma}
\begin{proof}
To save space, we only prove that $ \sup_{\theta_i\in \mathcal{B}_{i,\delta} }  | S_{T,i}^{\ast \eta_i \eta_i \eta_i}(\theta_i) | =\bar{O}_P(1)$. The proofs of the other results are similar. Define $\varrho_{\tau}^{(j)}(u) = \partial^j \varrho_{\tau}(u)/ \partial u^j  $ and $ g^{(j)} = \partial ^j g(u)/\partial u^j$. Then we can write
\begin{multline*}
S_{T,i}^{\ast \eta_i \eta_i \eta_i}(\theta_i) = -\frac{1}{T}\sum_{t=1}^{T}\varrho_{\tau}^{(3)}(Y_{it} -\theta_i' W_{it}) K_{it}/h^d    \\
= -\frac{3 }{T}\sum_{t=1}^{T}g^{(1)}\left( \frac{Y_{it} -\theta_i' W_{it} }{b}\right) \cdot \frac{1}{b^2} \cdot \frac{K_{it}}{h^d} -\frac{1 }{T}\sum_{t=1}^{T}g^{(2)}\left( \frac{Y_{it} -\theta_i' W_{it} }{b}\right) \cdot \frac{Y_{it} -\theta_i' W_{it}}{b^3} \cdot \frac{K_{it}}{h^d} .
\end{multline*}
Thus, we have
\begin{multline*}
\bar{S}_{T,i}^{\ast \eta_i \eta_i \eta_i}(\theta_i) = \Ex\left[ S_{T,i}^{\ast \eta_i \eta_i \eta_i}(\theta_i)\right] = -3 \Ex\left[ g^{(1)}\left( \frac{Y_{it} -\theta_i' W_{it} }{b}\right) \cdot \frac{1}{b^2} \cdot \frac{K_{it}}{h^d}\right] \\- \Ex\left[ g^{(2)}\left( \frac{Y_{it} -\theta_i' W_{it} }{b}\right) \cdot \frac{Y_{it} -\theta_i' W_{it}}{b^3} \cdot \frac{K_{it}}{h^d}\right].
\end{multline*}
First,
\begin{eqnarray*}
&&\Ex\left[ g^{(1)}\left( \frac{Y_{it} -\theta_i' W_{it} }{b}\right) \cdot \frac{1}{b^2} \cdot \frac{K_{it}}{h^d}\right]  \\
&=& \Ex\left[ \int g^{(1)} \left( \frac{u -(\theta_i-\theta_{0i})' W_{it} + m(X_{it}) }{b}\right) \frac{1}{b^2} f_{u,i}(u|X_{it})du \cdot \frac{K_{it}}{h^d} \right] \\
&= &-  \Ex\left[ \int g(v )f_{u,i}^{(1)} \left(vb + (\theta_i-\theta_{0i})' W_{it} - m(X_{it}) |X_{it} \right)dv\cdot \frac{K_{it}}{h^d} \right]  \\
&= &-  \Ex\left[ f_{u,i}^{(1)} \left( (\theta_i-\theta_{0i})' W_{it} - m(X_{it}) |X_{it} \right)\cdot \frac{K_{it}}{h^d} \right] +\bar{O}(b^m) \\
& = &-  \Ex\left[ f_{u,i}^{(1)} \left( (\theta_i-\theta_{0i})' W_{it}  |X_{it} \right)\cdot \frac{K_{it}}{h^d} \right]  + \bar{O}(b^m)+\bar{O}(h^2) =\bar{O}(1).
\end{eqnarray*}

Second,
\begin{eqnarray*}
&& \Ex\left[ g^{(2)}\left( \frac{Y_{it} -\theta_i' W_{it} }{b}\right) \cdot \frac{Y_{it} -\theta_i' W_{it}}{b^3} \cdot \frac{K_{it}}{h^d}\right]\\
&=& \Ex\left[ \int g^{(2)} \left( \frac{u -(\theta_i-\theta_{0i})' W_{it} + m(X_{it}) }{b}\right) \frac{u -(\theta_i-\theta_{0i})' W_{it} + m(X_{it})}{b^3} f_{u,i}(u|X_{it})du \cdot \frac{K_{it}}{h^d} \right] \\
&= & \Ex\left[ \int g^{(2)} \left( v\right) \frac{v}{b} f_{u ,i}\left(vb + (\theta_i-\theta_{0i})' W_{it} - m(X_{it}) |X_{it} \right)dv \cdot \frac{K_{it}}{h^d} \right]  \\
&=& - \Ex\left[ \int g^{(1)} \left( v\right) v f_{u ,i}^{(1)}\left(vb + (\theta_i-\theta_{0i})' W_{it} - m(X_{it}) |X_{it} \right)dv \cdot \frac{K_{it}}{h^d} \right]\\
&&  -\Ex\left[ \int g^{(1)} \left( v\right) 1/b  f_{u ,i}\left(vb + (\theta_i-\theta_{0i})' W_{it} - m(X_{it}) |X_{it} \right)dv \cdot \frac{K_{it}}{h^d} \right] \\
&=& b\cdot \Ex\left[ \int g\left( v\right) v f_{u ,i}^{(2)}\left(vb + (\theta_i-\theta_{0i})' W_{it} - m(X_{it}) |X_{it} \right)dv \cdot \frac{K_{it}}{h^d} \right]  -2 \Ex\left[ g^{(1)}\left( \frac{Y_{it} -\theta_i' W_{it} }{b}\right) \cdot \frac{1}{b^2} \cdot \frac{K_{it}}{h^d}\right]   \\
&=&\bar{O}(b^m)  -2 \Ex\left[ g^{(1)}\left( \frac{Y_{it} -\theta_i' W_{it} }{b}\right) \cdot \frac{1}{b^2} \cdot \frac{K_{it}}{h^d}\right] .
\end{eqnarray*}
It then follows that
\begin{multline}\label{A22} \bar{S}_{T,i}^{\ast \eta_i \eta_i \eta_i}(\theta_i) = - \Ex\left[ g^{(1)}\left( \frac{Y_{it} -\theta_i' W_{it} }{b}\right) \cdot \frac{1}{b^2} \cdot \frac{K_{it}}{h^d}\right] +\bar{O}(b^m)+\bar{O}(h^2) \\
=\Ex\left[ f_{u,i}^{(1)} \left( (\theta_i-\theta_{0i})' W_{it}  |X_{it} \right)\cdot \frac{K_{it}}{h^d} \right]  + \bar{O}(b^m)+\bar{O}(h^2) =\bar{O}(1).
\end{multline}

Third, consider
\[S_{T,i}^{\ast \eta_i \eta_i \eta_i}(\theta_i) - \bar{S}_{T,i}^{\ast \eta_i \eta_i \eta_i}(\theta_i)  =-\frac{1}{T}\sum_{t=1}^{T} \left\{ \varrho_{\tau}^{(3)}(Y_{it} -\theta_i' W_{it}) K_{it}/h^d  - \Ex\left[\varrho_{\tau}^{(3)}(Y_{it} -\theta_i' W_{it}) K_{it}/h^d\right] \right\}.\]
Similar to Lemma B.2 of \citet{galvao2016smoothed}, we can show that
\begin{equation}\label{A23}
    \max_{1\leq i\leq N}\sup_{\theta_i \in\mathcal{B}_{i,\delta} }\left\| S_{T,i}^{\ast \eta_i \eta_i \eta_i}(\theta_i) - \bar{S}_{T,i}^{\ast \eta_i \eta_i \eta_i}(\theta_i)  \right\| = o_P\left( \frac{\log N}{ \sqrt{Tb^3 h^d}}\right)
\end{equation}

Finally, the desired result follows from \eqref{A22}, \eqref{A23} and $\log N/\sqrt{Tb^3 h^d}\rightarrow 0$.
\end{proof}

\vspace{0.5cm}

\begin{lemma} Under Assumptions B1 to B4, we have
\[
S_{T,i}^{\ast \eta_i  \eta_i} =c_0 f_i(0)+\bar{o}_P(\log N/\sqrt{Th^d b}) +\bar{O}_P(h+b^m),\quad S_{T,i}^{\ast \phi  \eta_i} =\mathcal{C}_1 f_i(0)+\bar{o}_P(\log N/\sqrt{Th^d b}) +\bar{O}_P(h+b^m)\]
\[\quad S_{T,i}^{\ast \phi \phi}=\mathcal{C}_2 f_i(0)+\bar{o}_P(\log N/\sqrt{Th^d b}) +\bar{O}_P(h+b^m) \]
\end{lemma}
\begin{proof}
The proof is similar to the previous lemma, therefore it is omitted.
\end{proof}

\vspace{0.5cm}

\begin{lemma} Under Assumptions B1 to B4, we have
\[
\| \check{\phi} -\phi_0\|= O_P(1/\sqrt{NTh^d}) +O_P(h^2+b^m)  +o_P\left(  \max_{i\leq N} |\check{\eta}_i -\eta_{0i}| \right).
\]
\end{lemma}

\begin{proof}
Expanding the first order conditions we have
\begin{multline}\label{A24}
0=S_{NT}^{\ast \phi}(\check{\theta})=S_{NT}^{\ast \phi}+S_{NT}^{\ast \phi \phi} \cdot (\check{\phi} - \phi_0) + N^{-1}\sum_{i=1}^{N} S_{T,i}^{\ast \phi \eta_i}\cdot(\check{\eta}_i - \eta_{0i}) +0.5 \sum_{j=1}^d S_{NT}^{\ast \phi \phi \phi_j}(\bar{\theta})(\check{\phi} - \phi_0)(\check{\phi}_{j} - \phi_{0j}) +\\ 0.5  N^{-1}\sum_{i=1}^{N}S_{T,i}^{\ast \phi \phi \eta_i}(\bar{\theta}_i)(\check{\phi} - \phi_0)(\check{\eta}_{i} - \eta_{0i})  + 0.5  N^{-1}\sum_{i=1}^{N}S_{T,i}^{\ast \phi \eta_i \eta_i}(\bar{\theta}_i)(\check{\eta}_{i} - \eta_{0i})^2+ \\0.5N^{-1} \sum_{j=1}^d \sum_{i=1}^{N}S_{T,i}^{\ast \phi \eta_i \phi_j}(\bar{\theta}_i)(\check{\eta}_i - \eta_{0i})(\check{\phi}_{j} - \phi_{0j}),
\end{multline}
\begin{multline}\label{A25}
0=S_{T,i}^{\ast \eta_i}(\check{\theta}_i)=S_{T,i}^{\ast  \eta_i}+S_{T,i}^{\ast \eta_i \eta_i} \cdot (\check{\eta}_i - \eta_{0i}) + S_{T,i}^{\ast \eta_i \phi}\cdot(\check{\phi} - \phi_0) +0.5 \sum_{j=1}^d S_{T,i}^{\ast \eta_i \phi \phi_j}(\bar{\theta}_i)(\check{\phi} - \phi_0)(\check{\phi}_{j} - \phi_{0j}) +\\ 0.5  S_{T,i}^{\ast \eta_i \phi \eta_i}(\bar{\theta}_i)(\check{\phi} - \phi_0)(\check{\eta}_{i} - \eta_{0i})   + 0.5  S_{T,i}^{\ast \eta_i \eta_i \eta_i}(\bar{\theta}_i)(\check{\eta}_{i} - \eta_{0i})^2+0.5 \sum_{j=1}^d S_{T,i}^{\ast \eta_i \eta_i \phi_j}(\bar{\theta}_i)(\check{\eta}_i - \eta_{0i})(\check{\phi}_{j} - \phi_{0j}),
\end{multline}
where $\bar{\theta}_i$ is between $\theta_{0i}$ and $\check{\theta}_i$, and $\bar{\theta} = (\bar{\theta}_1,\ldots,\bar{\theta}_N)$.
It then follows from Lemma 4 to Lemma 6 and \eqref{A24}, \eqref{A25} that
\begin{equation}\label{A26}
\mathcal{C}_2 \bar{f}_N(0)\cdot (\check{\phi} - \phi_0)  =-S_{NT}^{\ast \phi} -\frac{\mathcal{C}_1}{N}\sum_{i=1}^{N}f_i(0)(\check{\eta}_i - \eta_{0i}) +o_P(\|\check{\phi} -\phi_0 \|) +o_P\left(  \max_{i\leq N}|\check{\eta}_i - \eta_{0i}|\right).
\end{equation}
\begin{equation}\label{A27}
c_0f_i(0) \cdot (\check{\eta}_i - \eta_{0i})  =-S_{T,i}^{\ast \eta_i} - f_i(0)\mathcal{C}_1'(\check{\phi} - \phi_{0}) +o_P(\|\check{\phi} -\phi_0 \|) +o_P\left(  \max_{i\leq N}|\check{\eta}_i - \eta_{0i}|\right).
\end{equation}
Plugging \eqref{A26} into \eqref{A27} gives
\begin{equation}\label{A28}
\mathcal{C} \bar{f}_N(0)\cdot (\check{\phi} - \phi_0)  = -\Bigg[ S_{NT}^{\ast \phi} -  \frac{\mathcal{C}_1}{c_0} \frac{1}{N}\sum_{i=1}^{N}S_{T,i}^{\ast \eta_i }\Bigg]+o_P(\|\check{\phi} -\phi_0 \|) +o_P\left(  \max_{i\leq N}|\check{\eta}_i - \eta_{0i}|\right).
\end{equation}
Write $Z_{it}^{\ast}= ((X_{it}-x)/h - \mathcal{C}_1/c_0 )\cdot K_{it}/h^d$, we have
\[S_{NT}^{\ast \phi} -  \frac{\mathcal{C}_1}{c_0} \frac{1}{N}\sum_{i=1}^{N}S_{T,i}^{\ast \eta_i }
=-  \frac{1}{NT} \sum_{i=1}^{N}\sum_{t=1}^{T}\varrho^{(1)}_{\tau}(Y_{it} -\theta_{0i}'W_{it}) Z_{it}^{\ast},
\]
where $\varrho^{(1)}_{\tau}(u)= \tau -G(u/b) + g(u/b)u/b$. We can write
\[S_{NT}^{\ast \phi} -  \frac{\mathcal{C}_1}{c_0} \frac{1}{N}\sum_{i=1}^{N}S_{T,i}^{\ast \eta_i } = -\frac{1}{NT} \sum_{i=1}^{N}\sum_{t=1}^{T}{\varrho}^{(1)}_{\tau}(u_{it}) Z_{it}^{\ast}
- \frac{1}{NT} \sum_{i=1}^{N}\sum_{t=1}^{T}[{\varrho}^{(1)}_{\tau}(Y_{it} -\theta_{0i}'W_{it})-{\varrho}^{(1)}_{\tau}(u_{it}) ]  Z_{it}^{\ast} .\]
The first term on the right-hand side of the above equation is $O_P(1/\sqrt{NTh^d}) +O(b^m)$ by the proof of the next lemma. Next, we focus on the second term on the RHS of the above equation, which can be written as:
\begin{multline}\label{A29}
-\frac{1}{NT} \sum_{i=1}^{N}\sum_{t=1}^{T}{\varrho}^{(2)}_{\tau}(u_{it} )  \left(0.5 (X_{it}-x)' \ddot{q}_{\tau}(x)(X_{it}-x)  + R_{\tau}(x,X_{it}) \right)Z_{it}^{\ast} \\
-\frac{1}{NT} \sum_{i=1}^{N}\sum_{t=1}^{T}{\varrho}^{(3)}_{\tau}[ u_{it} +c0.5 (X_{it}-x)' \ddot{q}_{\tau}(x)(X_{it}-x)  + cR_{\tau}(x,X_{it})]  \left[0.5 (X_{it}-x)' \ddot{q}_{\tau}(x)(X_{it}-x)  + R_{\tau}(x,X_{it})\right]^2Z_{it}^{\ast},
\end{multline}
where $c\in[0,1]$ and we have used the identity: $Y_{it} =   \theta_{0i}'W_{it} + 0.5 (X_{it}-x)' \ddot{q}_{\tau}(x)(X_{it}-x)  + R_{\tau}(x,X_{it})+u_{it}.$

First,
\begin{multline}\label{A30}
\frac{ 1}{NT} \sum_{i=1}^{N}\sum_{t=1}^{T}{\varrho}^{(2)}_{\tau}(u_{it} )  (X_{it}-x)' \ddot{q}_{\tau}(x)(X_{it}-x)  Z_{it}^{\ast} =
\frac{ 1}{N} \sum_{i=1}^{N}\Ex[ {\varrho}^{(2)}_{\tau}(u_{it} )  (X_{it}-x)' \ddot{q}_{\tau}(x)(X_{it}-x)  Z_{it}^{\ast}] + \\
\frac{ 1}{NT} \sum_{i=1}^{N}\sum_{t=1}^{T}\Big[ {\varrho}^{(2)}_{\tau}(u_{it} )  (X_{it}-x)' \ddot{q}_{\tau}(x)(X_{it}-x)  Z_{it}^{\ast}-\Ex[ {\varrho}^{(2)}_{\tau}(u_{it} )  (X_{it}-x)' \ddot{q}_{\tau}(x)(X_{it}-x)  Z_{it}^{\ast}]  \Big].
\end{multline}
For the first term on the RHS of \eqref{A30} we have:
\begin{multline}\label{A31} \Ex \big[ {\varrho}^{(2)}_{\tau}(u_{it} )  (X_{it}-x)' \ddot{q}_{\tau}(x)(X_{it}-x) [(X_{it}-x)/h-\mathcal{C}_1'/c_0] K_{it}/h^d\big]
\\=h^2 f_i(0) \int_{\mathcal{B}} u'\ddot{q}_{\tau}(0)u(u-\mathcal{C}_1/c_0)K(u)du +\bar{O}_P(h^3) = \bar{O}_P(h^2),
\end{multline}
and the second term can be shown to be $O_P(h^2/\sqrt{NTh^db})=o_P(1/\sqrt{NTh^d})$.

Second, we can show that
\[ \frac{1}{NT} \sum_{i=1}^{N}\sum_{t=1}^{T}{\varrho}^{(2)}_{\tau}(u_{it} )  R_{\tau}(x,X_{it})Z_{it}^{\ast} =O_P(h^3)+o_P(1/\sqrt{NTh^d}). \]

Third, we can show in a similar way that the second term of \eqref{A29} is $o_P(h^3)+O_P(h^4/\sqrt{NTh^db^3}) = o_P(h^3) +o_P(1/\sqrt{NTh^d}).$

Combining the results above, we have
\[\frac{1}{NT} \sum_{i=1}^{N}\sum_{t=1}^{T}[{\varrho}^{(1)}_{\tau}(Y_{it} -\theta_{0i}'W_{it})-{\varrho}^{(1)}_{\tau}(u_{it}) ]  Z_{it}^{\ast}  =O_P(h^2)+o_P(1/\sqrt{NTh^d})  \]
and
\[ S_{NT}^{\ast \phi} -  \frac{\mathcal{C}_1}{c_0} \frac{1}{N}\sum_{i=1}^{N}S_{T,i}^{\ast \eta_i }  =O_P(h^2+b^m)+O_P(1/\sqrt{NTh^d}) .\]
Then the desired result follows from \eqref{A28}.\end{proof}


\vspace{0.5cm}

\begin{lemma} Under Assumptions B1 to B4, we have
\[
\max_{i\leq N} |\check{\eta}_i -\eta_{0i} |= O_P(\sqrt{\log N}/ \sqrt{Th^d}).
\]
\end{lemma}
\begin{proof}
From Lemma 7 and \eqref{A27}
\begin{equation*}
c_0f_i(0) (\check{\eta}_i -\eta_{0i}) =-S_{T,i}^{\ast \eta_i }+ \bar{O}_P(1/\sqrt{NTh^d}) +\bar{O}_P(h^2+b^m)  +\bar{o}_P\left(  \max_{i\leq N} |\check{\eta}_i -\eta_{0i}| \right),
\end{equation*}
it then suffices to show that
\begin{equation}\label{A32}
\max_{i\leq N} | S_{T,i}^{\ast \eta_i }|  =\max_{i\leq N} \Bigg| \frac{1}{T} \sum_{t=1}^{T}{\varrho}^{(1)}_{\tau}(Y_{it} -\theta_{0i}'W_{it}) K_{it}/h^d \Bigg|  =O_P(\sqrt{\log N}/ \sqrt{Th^d})+O_P(h^2+b^m).
\end{equation}
Write
\begin{multline*}
 \frac{1}{T} \sum_{t=1}^{T}{\varrho}^{(1)}_{\tau}(Y_{it} -\theta_{0i}'W_{it}) K_{it}/h^d =
\frac{1}{T} \sum_{t=1}^{T} \left\{ {\varrho}^{(1)}_{\tau}(Y_{it} -\theta_{0i}'W_{it}) K_{it}/h^d -\Ex\left[{\varrho}^{(1)}_{\tau}(Y_{it} -\theta_{0i}'W_{it}) K_{it}/h^d \right] \right\}  \\+ \left\{ \Ex\left[ {\varrho}^{(1)}_{\tau}(Y_{it} -\theta_{0i}'W_{it}) K_{it}/h^d \right] - \Ex\left[ {\varrho}^{(1)}_{\tau}(u_{it}) K_{it}/h^d \right]  \right\}+\Ex\left[ {\varrho}^{(1)}_{\tau}(u_{it}) K_{it}/h^d \right]
\end{multline*}
From Lemma 2.2.9 and Lemma 2.2.10 of \citet{vanweak} it can shown that the first term on the RHS of the above equation is $\bar{O}_P(\sqrt{\log N}/ \sqrt{Th^d})$, and similar to the proof of Lemma 3 we can show that the second term is $\bar{O}_P(h^2)$. For the last term on the RHS of the above equation, we have
\begin{eqnarray*}
&& \Ex\left[ {\varrho}^{(1)}_{\tau}(u_{it}) K_{it}/h^d \right] \\
&=& \tau  \Ex\left[ K_{it}/h^d \right] - \Ex\left[ G(u_{it}/b) K_{it}/h^d \right] + \Ex\left[ g(u_{it}/b)u_{it}/b \cdot K_{it}/h^d \right] \\
&=&\tau  \Ex\left[ K_{it}/h^d \right]  - \Ex \left[ \int G(u/b) f_{u,i}(u|X_{it}) du\cdot K_{it}/h^d \right] + \Ex \left[ \int g(u/b)u/b\cdot f_{u,i}(u|X_{it}) du\cdot K_{it}/h^d \right]\\
&= &\tau  \Ex\left[ K_{it}/h^d \right]  - \Ex \left[ \int g(v) F_{u,i}(vb|X_{it}) dv\cdot K_{it}/h^d \right]+ b\Ex \left[ \int g(v)v\cdot f_{u,i}(vb|X_{it}) dv\cdot K_{it}/h^d \right] \\
&= &\tau  \Ex\left[ K_{it}/h^d \right]-\tau  \Ex\left[ K_{it}/h^d \right] +\bar{O}(b^m) =\bar{O}(b^m).
\end{eqnarray*}
Then the desired result follows. \end{proof}

\vspace{0.5cm}

\begin{lemma} Under Assumptions B1 to B4, we have
\begin{multline*}
\bar{f}_N(0) \mathcal{C}\cdot (\check{\phi} - \phi_0) = \frac{1}{NT} \sum_{i=1}^{N}\sum_{t=1}^{T}{\varrho}^{(1)}_{\tau}(u_{it})\bigg[  \frac{X_{it}-x}{h}
 - \mathcal{C}_1/c_0\bigg] \frac{K_{it}}{h^d} +h^2 B^{(1)} +O(h^3) \\+ N^{-1}\sum_{i=1}^{N} [ S_{T,i}^{\ast \phi \eta_i} - \mathcal{C}_1 S_{T,i}^{\ast \eta_i  \eta_i} /c_0]S_{T,i}^{\ast \eta_i }/(c_0f_i(0))
+o_P(\|\check{\phi} - \phi_0\|)+o_P(1/\sqrt{NTh^d}).
\end{multline*}
\end{lemma}

\begin{proof}
Plugging \eqref{A25} into \eqref{A24} we get:
\begin{multline} \label{A33}
\bar{f}_N(0) \mathcal{C} \cdot (\check{\phi} - \phi_0)
= -\Bigg[ S_{NT}^{\ast \phi} -  \frac{\mathcal{C}_1}{c_0 N}\sum_{i=1}^{N}S_{T,i}^{\ast \eta_i }\Bigg]  - N^{-1}\sum_{i=1}^{N} [ S_{T,i}^{\ast \phi \eta_i} -\mathcal{C}_1 S_{T,i}^{\ast \eta_i  \eta_i} /c_0](\check{\eta}_i - \eta_{0i})   \\
-0.5  N^{-1}\sum_{i=1}^{N}\Bigg[ S_{T,i}^{\ast \phi \eta_i \eta_i}(\bar{\theta}_i)-\mathcal{C}_1 S_{T,i}^{\ast \eta_i \eta_i \eta_i}(\bar{\theta}_i) /c_0\Bigg](\check{\eta}_{i} - \eta_{0i})^2 + o_P(\|\check{\phi} - \phi_0\|)
\end{multline}

First, from the proof of Lemma 7 we have
\begin{multline}\label{A34}
S_{NT}^{\ast \phi} -  \frac{\mathcal{C}_1}{c_0 N}\sum_{i=1}^{N}S_{T,i}^{\ast \eta_i }
= -\frac{1}{NT} \sum_{i=1}^{N}\sum_{t=1}^{T}{\varrho}^{(1)}_{\tau}(u_{it})\bigg[  \frac{X_{it}-x}{h} - \mathcal{C}_1/c_0\bigg] \frac{K_{it}}{h^d} \\
-0.5h^2 \bar{f}_N(0) \int_{\mathcal{B}} u'\ddot{q}_{\tau}(0)u(u-\mathcal{C}_1/c_0)K(u)du +\bar{O}_P(h^3)+o_P(1/\sqrt{NTh^d}).
\end{multline}

Second, from the proof of Lemma 5 we have
\[
S_{T,i}^{\ast \phi \eta_i \eta_i}(\bar{\theta}_i) = f_{u,i}^{(1)}(0|0) f_{X,i}(0) \cdot \mathcal{C}_1 + \bar{O}_P\left( \|\bar{\theta}_i -\theta_{0i} \|\right)
+\bar{O}_P(h+b^m) +\bar{o}_P\left( \frac{\log N}{ \sqrt{Tb^3 h^d}}\right),
\]
\[
S_{T,i}^{\ast\eta_i  \eta_i \eta_i}(\bar{\theta}_i) = f_{u,i}^{(1)}(0|0) f_{X,i}(0) \cdot c_0 + \bar{O}_P\left( \|\bar{\theta}_i -\theta_{0i} \|\right)
+\bar{O}_P(h+b^m) +\bar{o}_P\left( \frac{\log N}{ \sqrt{Tb^3 h^d}}\right).
\]
It then follows that
\begin{equation*}
S_{T,i}^{\ast \phi \eta_i \eta_i}(\bar{\theta}_i)-\mathcal{C}_1 S_{T,i}^{\ast \eta_i \eta_i \eta_i}(\bar{\theta}_i) /c_0
=\bar{O}_P\left( \|\bar{\theta}_i -\theta_{0i} \|\right)
+\bar{O}_P(h+b^m) +\bar{o}_P\left( \frac{\log N}{ \sqrt{Tb^3 h^d}}\right) .
\end{equation*}
The above equation and Lemma 8 imply that
\begin{multline}
N^{-1}\sum_{i=1}^{N}\Bigg[ S_{T,i}^{\ast \phi \eta_i \eta_i}(\bar{\theta}_i)-\mathcal{C}_1 S_{T,i}^{\ast \eta_i \eta_i \eta_i}(\bar{\theta}_i) /c_0\Bigg](\check{\eta}_{i} - \eta_{0i})^2 = O_P\left( \max_{i\leq N} |\check{\eta}_i-\eta_{0i}|^3 \right) \\+O_P\left( \max_{i\leq N} |\check{\eta}_i-\eta_{0i}|^2 \right)\cdot O_P\left( h\right)
+ O_P\left( \max_{i\leq N} |\check{\eta}_i-\eta_{0i}|^2 \right) \cdot o_P\left( \frac{\log N}{ \sqrt{Tb^3 h^d}}\right)
\\ = O_P\left( \frac{(\log N)^{3/2}}{ (Th^d)^{3/2}}  \right) + O_P\left( \frac{h \log N }{ Th^d}  \right)  + o_P\left( \frac{\log N}{ Th^d} \cdot \frac{\log N}{ \sqrt{Th^d b^3}}  \right)
=o_P(1/\sqrt{NTh^d}).
\end{multline}

Third, by \eqref{A25}, Lemma 6 and Lemma 8 we have
\begin{multline*}
 c_0 f_i(0)(\check{\eta}_i - \eta_{0i} )= -S_{T,i}^{\ast \eta_i } - ( S_{T,i}^{\ast \eta_i \eta_i }-c_0 f_i(0))(\check{\eta}_i - \eta_{0i} )  +\bar{O}_P(\check{\phi}-\phi_0)+ \bar{O}_P\big(\log N / (Th^d)\big) \\
 = -S_{T,i}^{\ast \eta_i }+\bar{O}_P(\check{\phi}-\phi_0)+\bar{o}_P\left( \frac{ (\log N)^{3/2} }{Th^d\sqrt{b} }\right) +\bar{O}_P\left( \frac{ h\sqrt{\log N} }{\sqrt{Th^d} }\right)
\end{multline*}
Thus, it follows from the above result and Lemma 6 that
\begin{multline}\label{A36}
N^{-1}\sum_{i=1}^{N} [ S_{T,i}^{\ast \phi \eta_i} -\mathcal{C}_1 S_{T,i}^{\ast \eta_i  \eta_i} /c_0](\check{\eta}_i - \eta_{0i})
= - N^{-1}\sum_{i=1}^{N} [ S_{T,i}^{\ast \phi \eta_i} -\mathcal{C}_1 S_{T,i}^{\ast \eta_i  \eta_i} /c_0]S_{T,i}^{\ast \eta_i }/(c_0f_i(0)) \\
+o_P(\| \check{\phi}-\phi_0\|)+o_P\left( \frac{ (\log N)^{5/2} }{(Th^d)^{3/2}b }\right)+o_P\left( \frac{ h (\log N)^{3/2} }{Th^d \sqrt{b} }\right)
+O_P\left( \frac{ h^2\sqrt{\log N} }{\sqrt{Th^d} }\right)\\
=  - N^{-1}\sum_{i=1}^{N} [ S_{T,i}^{\ast \phi \eta_i} -\mathcal{C}_1 S_{T,i}^{\ast \eta_i  \eta_i} /c_0]S_{T,i}^{\ast \eta_i }/(c_0f_i(0))
+o_P(\| \check{\phi}-\phi_0\|) + o_P(\sqrt{NTh^d}).
\end{multline}

Finally, the desired result follows from \eqref{A33} to \eqref{A36}.\end{proof}

\vspace{0.5cm}

\begin{lemma} We have
\[N^{-1}\sum_{i=1}^{N} [ S_{T,i}^{\ast \phi \eta_i} - \mathcal{C}_1 S_{T,i}^{\ast \eta_i  \eta_i} /c_0]S_{T,i}^{\ast \eta_i }/(c_0f_i(0))  =- \frac{\tau-1/2}{Th^d} (\mathcal{D}_1/c_0 -\mathcal{C}_1 d_0/c_0^2 )+o_P((Th^d)^{-1}).\]
\end{lemma}
\begin{proof}

\textbf{Step 1:}

First, write
\begin{eqnarray*}
\Ex[ S_{T,i}^{\ast \phi \eta_i} S_{T,i}^{\ast \eta_i }  ]  &=&
 - \Ex\Bigg[ \Bigg( \frac{1}{T}\sum_{t=1}^{T} {\varrho}^{(2)}_{\tau}(Y_{it}-\theta_{0i}'W_{it})Z_{it} \Bigg)  \Bigg(\frac{1}{T}\sum_{t=1}^{T} {\varrho}^{(1)}_{\tau}(Y_{it}-\theta_{0i}'W_{it})K_{it}/h^d  \Bigg) \Bigg]  \\
 &=& - \frac{1}{Th^d}\Ex \Big[{\varrho}^{(2)}_{\tau}(Y_{it}-\theta_{0i}'W_{it})  {\varrho}^{(1)}_{\tau}(Y_{it}-\theta_{0i}'W_{it})[(X_{it}-x)/h] K_{it}^2/h^d \Big] \\
 &&-  \frac{T-1}{T} \Ex[{\varrho}^{(2)}_{\tau}(Y_{it}-\theta_{0i}'W_{it})Z_{it}  ]\cdot \Ex[{\varrho}^{(1)}_{\tau}(Y_{it}-\theta_{0i}'W_{it})K_{it}/h^d ].
\end{eqnarray*}

Second, it can be shown that
\[ \Ex \Big[{\varrho}^{(2)}_{\tau}(Y_{it}-\theta_{0i}'W_{it})  {\varrho}^{(1)}_{\tau}(Y_{it}-\theta_{0i}'W_{it})[(X_{it}-x)/h] K_{it}^2/h^d \Big]
=(\tau-1/2)f_i(0)\int_{\mathcal{B}} u K^2(u)du +\bar{o}(1).
\]
It then follows that
\begin{multline}\label{A37}
 \Ex[ S_{T,i}^{\ast \phi \eta_i} S_{T,i}^{\ast \eta_i }  ] = -\frac{1}{Th^d} (\tau-1/2)f_i(0)\int_{\mathcal{B}} u K^2(u)du \\-  \frac{T-1}{T} \Ex[{\varrho}^{(2)}_{\tau}(Y_{it}-\theta_{0i}'W_{it})Z_{it}  ]\cdot \Ex[{\varrho}^{(1)}_{\tau}(Y_{it}-\theta_{0i}'W_{it})K_{it}/h^d ]+\bar{o}( (Th^d)^{-1}).
\end{multline}
Similarly, we can show that
\begin{multline}\label{A38}
 \Ex[ S_{T,i}^{\ast \eta_i \eta_i} S_{T,i}^{\ast \eta_i }  ] = -\frac{1}{Th^d} (\tau-1/2)f_i(0)\int_{\mathcal{B}}  K^2(u)du\\-  \frac{T-1}{T} \Ex[{\varrho}^{(2)}_{\tau}(Y_{it}-\theta_{0i}'W_{it})K_{it}/h^d]\cdot \Ex[{\varrho}^{(1)}_{\tau}(Y_{it}-\theta_{0i}'W_{it})K_{it}/h^d ]+\bar{o}((Th^d)^{-1}).
\end{multline}

Third, it follows from \eqref{A37} and \eqref{A38} that
\begin{multline}\label{A39}
\Ex\Bigg[ N^{-1}\sum_{i=1}^{N} [ S_{T,i}^{\ast \phi \eta_i} - \mathcal{C}_1 S_{T,i}^{\ast \eta_i  \eta_i} /c_0]S_{T,i}^{\ast \eta_i }/(c_0f_i(0)) \Bigg]= - \frac{\tau-1/2}{Th^d} (\mathcal{D}_1/c_0 -\mathcal{C}_1 d_0/c_0^2 )+o((Th^d)^{-1}) \\
-  \frac{T-1}{T}   \frac{1}{N}\sum_{i=1}^{N}\Big[ \Ex[{\varrho}^{(2)}_{\tau}(Y_{it}-\theta_{0i}'W_{it})Z_{it}  ] -\Ex[{\varrho}^{(2)}_{\tau}(Y_{it}-\theta_{0i}'W_{it})K_{it}/h^d  ] \mathcal{C}_1 /c_0 \Big]  \cdot \Ex[{\varrho}^{(1)}_{\tau}(Y_{it}-\theta_{0i}'W_{it})K_{it}/h^d ]/(c_0f_i(0)),
\end{multline}
and we can show that $\Ex[{\varrho}^{(2)}_{\tau}(Y_{it}-\theta_{0i}'W_{it})Z_{it}  ] =\mathcal{C}_1f_i(0)+\bar{O}(h)$, $\Ex[{\varrho}^{(2)}_{\tau}(Y_{it}-\theta_{0i}'W_{it})K_{it}/h^d] =c_0f_i(0) +\bar{O}(h)$ and $\Ex[{\varrho}^{(1)}_{\tau}(Y_{it}-\theta_{0i}'W_{it})K_{it}/h^d ] = \bar{O}(h^2)$. It then follows that
\begin{equation}\label{A40} \Big[ \Ex[{\varrho}^{(2)}_{\tau}(Y_{it}-\theta_{0i}'W_{it})Z_{it}  ] -\Ex[{\varrho}^{(2)}_{\tau}(Y_{it}-\theta_{0i}'W_{it})K_{it}/h^d  ] \mathcal{C}_1 /c_0 \Big]  \cdot \Ex[{\varrho}^{(1)}_{\tau}(Y_{it}-\theta_{0i}'W_{it})K_{it}/h^d ] = \bar{O}(h^3) =\bar{o}((Th^d)^{-1}).
\end{equation}

Finally, it follows from \eqref{A39} and \eqref{A40} that
\begin{equation}\label{A41}
\Ex\Bigg[ N^{-1}\sum_{i=1}^{N} [ S_{T,i}^{\ast \phi \eta_i} - \mathcal{C}_1 S_{T,i}^{\ast \eta_i  \eta_i} /c_0]S_{T,i}^{\ast \eta_i }/(c_0f_i(0)) \Bigg]= - \frac{\tau-1/2}{Th^d} (\mathcal{D}_1/c_0 -\mathcal{C}_1 d_0/c_0^2 )+o((Th^d)^{-1}).
\end{equation}

\textbf{Step 2:}
Now we will show that
\begin{equation}\label{A42}
\text{Var}\Bigg[ N^{-1}\sum_{i=1}^{N} [ S_{T,i}^{\ast \phi \eta_i} - \mathcal{C}_1 S_{T,i}^{\ast \eta_i  \eta_i} /c_0]S_{T,i}^{\ast \eta_i }/(c_0f_i(0)) \Bigg]= o(1/(Th^d)^2).
\end{equation}
Then the desired result follows from \eqref{A41} and \eqref{A42}.

Define $\tilde{S}_{T,i}^{\ast \phi \eta_i} =S_{T,i}^{\ast \phi \eta_i} - \Ex[S_{T,i}^{\ast \phi \eta_i}]$, and $\tilde{S}_{T,i}^{\ast \eta_i \eta_i} =S_{T,i}^{\ast \eta_i \eta_i} - \Ex[S_{T,i}^{\ast \eta_i \eta_i}]$,
then we can write
\begin{multline*}
N^{-1}\sum_{i=1}^{N} [ S_{T,i}^{\ast \phi \eta_i} - \mathcal{C}_1 S_{T,i}^{\ast \eta_i  \eta_i} /c_0]S_{T,i}^{\ast \eta_i }/(c_0f_i(0))
= N^{-1}\sum_{i=1}^{N} [ \tilde{S}_{T,i}^{\ast \phi \eta_i} - \mathcal{C}_1 \tilde{S}_{T,i}^{\ast \eta_i  \eta_i} /c_0]S_{T,i}^{\ast \eta_i }/(c_0f_i(0)) \\+ N^{-1}\sum_{i=1}^{N} [ \bar{S}_{T,i}^{\ast \phi \eta_i} - \mathcal{C}_1 \bar{S}_{T,i}^{\ast \eta_i  \eta_i} /c_0]S_{T,i}^{\ast \eta_i }/(c_0f_i(0)),
 \end{multline*}
and it follows that for any $\omega\in\mathbb{R}^d$,
\begin{multline*} \text{Var}\Bigg[ N^{-1}\sum_{i=1}^{N}\omega' [ S_{T,i}^{\ast \phi \eta_i} - \mathcal{C}_1 S_{T,i}^{\ast \eta_i  \eta_i} /c_0]S_{T,i}^{\ast \eta_i }/(c_0f_i(0)) \Bigg]
\leq 2  \text{Var}\Bigg[ N^{-1}\sum_{i=1}^{N}\omega' [ \tilde{S}_{T,i}^{\ast \phi \eta_i} - \mathcal{C}_1 \tilde{S}_{T,i}^{\ast \eta_i  \eta_i} /c_0]S_{T,i}^{\ast \eta_i }/(c_0f_i(0)) \Bigg]  \\
+ 2\text{Var}\Bigg[ N^{-1}\sum_{i=1}^{N} \omega'[ \bar{S}_{T,i}^{\ast \phi \eta_i} - \mathcal{C}_1 \bar{S}_{T,i}^{\ast \eta_i  \eta_i} /c_0]S_{T,i}^{\ast \eta_i }/(c_0f_i(0)) \Bigg] .
\end{multline*}
First, we have
\begin{eqnarray*}
&& \text{Var}\Bigg[ N^{-1}\sum_{i=1}^{N} \omega'[ \tilde{S}_{T,i}^{\ast \phi \eta_i} - \mathcal{C}_1 \tilde{S}_{T,i}^{\ast \eta_i  \eta_i} /c_0]S_{T,i}^{\ast \eta_i }/(c_0f_i(0)) \Bigg]  \\
&=& N^{-2}\sum_{i=1}^{N}\text{Var} \bigg\{      \omega'[ \tilde{S}_{T,i}^{\ast \phi \eta_i} - \mathcal{C}_1 \tilde{S}_{T,i}^{\ast \eta_i  \eta_i} /c_0    ]S_{T,i}^{\ast \eta_i }  \bigg\}/(c_0f_i(0))^2 \\
&\leq & N^{-2}\sum_{i=1}^{N}\Ex \bigg\{      \omega'[ \tilde{S}_{T,i}^{\ast \phi \eta_i} - \mathcal{C}_1 \tilde{S}_{T,i}^{\ast \eta_i  \eta_i} /c_0    ]S_{T,i}^{\ast \eta_i }  \bigg\}^2/(c_0f_i(0))^2.
\end{eqnarray*}
Note that
\begin{eqnarray*}
&&  \Ex \bigg\{      \omega'[ \tilde{S}_{T,i}^{\ast \phi \eta_i} - \mathcal{C}_1 \tilde{S}_{T,i}^{\ast \eta_i  \eta_i} /c_0    ]S_{T,i}^{\ast \eta_i }  \bigg\}^2  \\
&=&  \Ex  \Bigg\{    \Bigg( \frac{1}{Th^d}\sum_{t=1}^{T}a_{it}  \Bigg)^2     \Bigg(\frac{1}{Th^d}\sum_{t=1}^{T} b_{it}\Bigg)^2    \Bigg\} = \frac{1}{(Th^d)^4 } \sum_{t=1}^{T}\sum_{s=1}^{T}\sum_{p=1}^{T}\sum_{h=1}^{T} \Ex[a_{it}a_{is}b_{ip}b_{ih}] \\
&=&\frac{1}{(Th^d)^4 } \sum_{t=1}^{T}\sum_{p=1}^{T}\sum_{h=1}^{T} \Ex[a_{it}^2b_{ip}b_{ih}]
+\frac{1}{(Th^d)^4 } \sum_{t=1}^{T}\sum_{h=1}^{T} \Ex[a_{it}^2 b_{ih}^2 ]
+\frac{1}{(Th^d)^4 } \sum_{t=1}^{T}\sum_{h=1}^{T} \Ex[a_{it}^2b_{it}b_{ih}] \\
&& +\frac{1}{(Th^d)^4 } \sum_{t=1}^{T}\Ex[a_{it}^2b_{it}^2]
+\frac{1}{(Th^d)^4 } \sum_{t=1}^{T}\sum_{s=1}^{T} \Ex[a_{it}a_{is}b_{it}b_{is}],
\end{eqnarray*}
where $ a_{it} = \omega' {\varrho}^{(2)}_{\tau}(Y_{it}-\theta_{0i}'W_{it})[(X_{it}-x)/h -\mathcal{C}_1/c_0] K_{it} - \Ex[ \omega' {\varrho}^{(2)}_{\tau}(Y_{it}-\theta_{0i}'W_{it})[(X_{it}-x)/h -\mathcal{C}_1/c_0] K_{it}]$, and $b_{it} ={\varrho}^{(1)}_{\tau}(Y_{it}-\theta_{0i}'W_{it})K_{it} $.
It can be shown that
\[ \frac{1}{(Th^d)^4 } \sum_{t=1}^{T}\sum_{p=1}^{T}\sum_{h=1}^{T} \Ex[a_{it}^2b_{ip}b_{ih}] =\bar{O}(h^4/(Th^d b)), \quad  \frac{1}{(Th^d)^4 } \sum_{t=1}^{T}\sum_{h=1}^{T} \Ex[a_{it}^2 b_{ih}^2 ] = \bar{O}(1/[(Th^d)^2b]),  \]
\[\frac{1}{(Th^d)^4 } \sum_{t=1}^{T}\sum_{h=1}^{T} \Ex[a_{it}^2b_{it}b_{ih}] = \bar{O}(h^2/[(Th^d)^2 b]) ,\]
\[ \frac{1}{(Th^d)^4 } \sum_{t=1}^{T}\Ex[a_{it}^2b_{it}^2] =\bar{O}(1/[(Th^d)^3b]) ,\quad  \frac{1}{(Th^d)^4 } \sum_{t=1}^{T}\sum_{s=1}^{T} \Ex[a_{it}a_{is}b_{it}b_{is}] =\bar{O}(1/[(Th^d)^2]) .\]
Thus, we have
\begin{equation}\label{A43}
 \text{Var}\Bigg[ N^{-1}\sum_{i=1}^{N} \omega'[ \tilde{S}_{T,i}^{\ast \phi \eta_i} - \mathcal{C}_1 \tilde{S}_{T,i}^{\ast \eta_i  \eta_i} /c_0]S_{T,i}^{\ast \eta_i }/(c_0f_i(0)) \Bigg]  =o(1/(NTh^d)).\end{equation}

Second, define $\zeta_i= \omega'[ \bar{S}_{T,i}^{\ast \phi \eta_i} - \mathcal{C}_1 \bar{S}_{T,i}^{\ast \eta_i  \eta_i} /c_0]/(c_0f_i(0))$,
then we can write
\[\text{Var}\Bigg[ N^{-1}\sum_{i=1}^{N} \omega'[ \bar{S}_{T,i}^{\ast \phi \eta_i} - \mathcal{C}_1 \bar{S}_{T,i}^{\ast \eta_i  \eta_i} /c_0]S_{T,i}^{\ast \eta_i }/(c_0f_i(0)) \Bigg] =
\frac{1}{(NTh^d)^2} \sum_{i=1}^{N}\sum_{t=1}^{T} \text{Var}[{\varrho}^{(1)}_{\tau}(Y_{it}-\theta_{0i}'W_{it})K_{it} ]\zeta_i^2.
\]
Since $\zeta_i=\bar{o}(1)$ and $\text{Var}[{\varrho}^{(1)}_{\tau}(Y_{it}-\theta_{0i}'W_{it})K_{it} ] = \bar{O}(h^d)$, it follows that
\begin{equation}\label{A44}
\text{Var}\Bigg[ N^{-1}\sum_{i=1}^{N} \omega'[ \bar{S}_{T,i}^{\ast \phi \eta_i} - \mathcal{C}_1 \bar{S}_{T,i}^{\ast \eta_i  \eta_i} /c_0]S_{T,i}^{\ast \eta_i }/(c_0f_i(0)) \Bigg]  =o(1/(NTh^d)).
\end{equation}
Finally, \eqref{A42} follows from \eqref{A43} and \eqref{A44}, and this concludes the proof.  \end{proof}

\vspace{0.5cm}
\noindent{\textbf{Proof of Theorem 2:}}
\begin{proof}
It follows from Lemma 10 and Lemma 11 that
\begin{multline*}
\check{\phi} - \phi_0 = \frac{\mathcal{C}^{-1}}{\bar{f}_N(0)NT} \sum_{i=1}^{N}\sum_{t=1}^{T}{\varrho}^{(1)}_{\tau}(u_{it})\bigg[  \frac{X_{it}-x}{h}
 - \mathcal{C}_1/c_0\bigg] \frac{K_{it}}{h^d}  +h^2 B^{(1)} \\
  + \frac{1}{ Th^d} B^{(2)}+o_P(\|\check{\phi} - \phi_0\|)+o_P(1/(Th^d)) +O(h^3).
\end{multline*}
It then suffices to show that
\begin{equation}\label{A45}
\frac{\mathcal{C}^{-1}}{\bar{f}_N(0) \sqrt{NT}} \sum_{i=1}^{N}\sum_{t=1}^{T}{\varrho}^{(1)}_{\tau}(u_{it})\bigg[  \frac{X_{it}-x}{h}
 - \mathcal{C}_1/c_0\bigg] \frac{K_{it}}{\sqrt{h^d}} \overset{d}{\rightarrow} \mathcal{N}(0, \tau(1-\tau)\sigma(0)\Omega ) .
\end{equation}
First, we can show that
\[ \Ex \Bigg[   {\varrho}^{(1)}_{\tau}(u_{it})\bigg[  \frac{X_{it}-x}{h}
 - \mathcal{C}_1/c_0\bigg] \frac{K_{it}}{\sqrt{h^d}} \Bigg] =\bar{O} (  h^{d/2}b^m )  =\bar{o} ( 1/ \sqrt{NT} ) . \]
Second, we can show that
\[ \text{Var}\Bigg[    {\varrho}^{(1)}_{\tau}(u_{it})\bigg[  \frac{X_{it}-x}{h}
 - \mathcal{C}_1/c_0\bigg] \frac{K_{it}}{\sqrt{h^d}}    \Bigg]
 =\tau(1-\tau) f_{X,i}(0) \int_{\mathcal{B}}(u-\mathcal{C}_1/c_0)(u-\mathcal{C}_1/c_0)' K^2(u)du + \bar{o} ( 1 ) ,
 \]
then \eqref{A45} follows from Lyapunov's CLT.
\end{proof}

\subsection{Tables}
\begin{table}[htp]
\small
\begin{center}
\begin{threeparttable}
\caption{Biases and MSEs of the Estimators at $\tau=0.25$ with Gaussian Errors.}
\begin{tabular}{cc|cccc|cccc}
\hline
  &  & \multicolumn{4}{c|}{Bias} &  \multicolumn{4}{c}{MSE} \\
  $x$ &   $\beta_{\tau}(x)$  & $\hat{\beta}_{\tau}$ &  $\hat{\beta}_{\tau}^{bc}$  & $\check{\beta}_{\tau}$ &  $\check{\beta}_{\tau}^{bc}$ & $\hat{\beta}_{\tau}$ &  $\hat{\beta}_{\tau}^{bc}$  & $\check{\beta}_{\tau}$ &  $\check{\beta}^{bc}_{\tau}$ \\
\hline
 -2.0&  1.603 & -0.762   &0.156 & -0.747&-0.092&0.934  &0.727&0.866&0.485 \\
 -1.6&  1.572&   -0.179 &  -0.015& -0.189& -0.023& 0.082 & 0.084&0.088& 0.098\\
  -1.2&   1.518&  -0.082  & -0.031& -0.084&-0.028& 0.021 &0.023&0.022& 0.026\\
 -0.8&  1.421 &   -0.077 & -0.052& -0.071 &-0.048& 0.013 & 0.013&0.011& 0.013\\
  -0.4&   1.251&   0.061 & -0.052& -0.061 &-0.051& 0.007 &0.008&0.008& 0.009\\
  0.0&  1.000 &   0.001 & 0.000& -0.005& -0.005&  0.003&0.004&0.003& 0.005\\
  0.4&   0.750&    0.060& 0.049 & 0.055&0.046 & 0.007 &0.007&0.007& 0.008\\
  0.8&  0.579&    0.076& 0.051& 0.074& 0.053& 0.012 &0.012&0.013& 0.015\\
  1.2&   0.482&    0.091& 0.040& 0.085&0.031& 0.022 &0.022&0.022& 0.028\\
 1.6&  0.428&  0.188  & 0.021& 0.183&0.021& 0.090 & 0.087&0.093& 0.113\\
  2.0&   0.397&   0.736 &-0.175 & 0.730&0.086& 0.872 &0.679&0.814& 0.464\\
  \hline
\end{tabular}

\begin{tablenotes}
      \small
      \item Note: The DGP considered in this table is: $Y_{it} = \beta X_{it} + \alpha_i + \sqrt{ 1+X_{it}^2} \cdot \epsilon_{it}$, $X_{it}\sim i.i.d \text{ } \mathcal{N}(0,1)\cdot \mathbf{1}\{|X_{it}|\leq 2\}$, $\alpha_i,\epsilon_{it}\sim i.i.d \text{ } \mathcal{N}(0,1)$, so $\beta_{\tau}(x) = 1+ \Phi^{-1}(\tau)\cdot x/\sqrt{1+x^2}$.
    \end{tablenotes}
\end{threeparttable}

\end{center}
\label{default}
\end{table}%

\begin{table}[htp]
\small
\begin{center}
\begin{threeparttable}
\caption{Biases and MSEs of the Estimators at $\tau=0.25$ with $T(3)$ Errors.}
\begin{tabular}{cc|cccc|cccc}
\hline
  &  & \multicolumn{4}{c|}{Bias} &  \multicolumn{4}{c}{MSE} \\
  $x$ &   $\beta_{\tau}(x)$  & $\hat{\beta}_{\tau}$ &  $\hat{\beta}_{\tau}^{bc}$  & $\check{\beta}_{\tau}$ &  $\check{\beta}_{\tau}^{bc}$ & $\hat{\beta}_{\tau}$ &  $\hat{\beta}_{\tau}^{bc}$  & $\check{\beta}_{\tau}$ &  $\check{\beta}^{bc}_{\tau}$ \\
\hline
 -2.0&  1.684 & -0.981   &0.150 & -0.915&-0.198&1.428  &0.896&1.308&0.715 \\
 -1.6&  1.649&   -0.241 &  -0.026& -0.230& -0.007& 0.130 & 0.125&0.129& 0.144\\
  -1.2&   1.588&  -0.112  & -0.041& -0.117&-0.046& 0.032 &0.033&0.037& 0.041\\
 -0.8&  1.478 &   -0.090 & -0.057& -0.091 &-0.058& 0.018 & 0.018&0.018& 0.019\\
  -0.4&   1.284&   -0.069 & -0.055& -0.069 &-0.057& 0.010 &0.011&0.010& 0.012\\
  0.0&  1.000 &   0.001 & 0.007& -0.004& -0.006&  0.005&0.007&0.005& 0.008\\
  0.4&   0.716&    0.071& 0.059 & 0.070&0.059 & 0.010 &0.011&0.010& 0.012\\
  0.8&  0.522&    0.095& 0.059& 0.097& 0.063& 0.019 &0.018&0.020& 0.021\\
  1.2&   0.412&    0.116& 0.046& 0.115&0.039& 0.033 &0.035&0.036& 0.042\\
 1.6&  0.351&  0.218  & -0.001& 0.219&0.004& 0.125 & 0.128&0.129& 0.151\\
  2.0&   0.316&   0.895 &-0.238 & 0.911&0.220& 1.253 &0.949&1.266& 0.701\\
  \hline
\end{tabular}

\begin{tablenotes}
      \small
      \item Note: The DGP considered in this table is: $Y_{it} = \beta X_{it} + \alpha_i + \sqrt{ 1+X_{it}^2} \cdot \epsilon_{it}$, $X_{it}\sim i.i.d \text{ } \mathcal{N}(0,1)\cdot \mathbf{1}\{|X_{it}|\leq 2\}$, $\alpha_i\sim i.i.d \text{ } \mathcal{N}(0,1)$, $\epsilon_{it}\sim T(3)$, so $\beta_{0.25}(x) = 1-0.765\cdot x/\sqrt{1+x^2}$.
    \end{tablenotes}
\end{threeparttable}

\end{center}
\label{default}
\end{table}%

\begin{table}[htp]
\small
\begin{center}
\begin{threeparttable}
\caption{Biases and MSEs of the Estimators at $\tau=0.5$ with Gaussian Errors.}
\begin{tabular}{cc|cccc|cccc}
\hline
  &  & \multicolumn{4}{c|}{Bias} &  \multicolumn{4}{c}{MSE} \\
  $x$ &   $\beta_{\tau}(x)$  & $\hat{\beta}_{\tau}$ &  $\hat{\beta}_{\tau}^{bc}$  & $\check{\beta}_{\tau}$ &  $\check{\beta}_{\tau}^{bc}$ & $\hat{\beta}_{\tau}$ &  $\hat{\beta}_{\tau}^{bc}$  & $\check{\beta}_{\tau}$ &  $\check{\beta}^{bc}_{\tau}$ \\
\hline
 -2.0&  1.000 & -0.002   &0.007 & -0.014&-0.009&0.236  &0.455&0.259&0.493 \\
 -1.6&  1.000&   0.005 &  0.015& 0.007& -0.020& 0.041 & 0.070&0.043& 0.079\\
  -1.2&   1.000&  0.007  & 0.006& 0.005&-0.006& 0.012 &0.018&0.013& 0.022\\
 -0.8&  1.000 &   -0.003 & -0.003& -0.005 &-0.006& 0.005 & 0.007&0.005& 0.008\\
  -0.4&   1.000&   -0.004 & -0.003& -0.006 &-0.007& 0.003 &0.005&0.003& 0.005\\
  0.0&  1.000 &   -0.005 & -0.006& -0.007& -0.009&  0.003&0.004&0.003& 0.004\\
  0.4&   1.000&    -0.003& -0.004 & -0.003&-0.004 & 0.003 &0.004&0.003& 0.004\\
  0.8&  1.000&    0.003& 0.005& 0.002& 0.003& 0.006 &0.008&0.006& 0.009\\
  1.2&   1.000&    0.002& 0.002& 0.001&0.003& 0.012 &0.017&0.012& 0.020\\
 1.6&  1.000&  -0.003  & -0.001& -0.006&-0.012& 0.047 & 0.077&0.047& 0.081\\
  2.0&   1.000&   0.002 & 0.039 & 0.001&0.032& 0.256 &0.447&0.273& 0.472\\
  \hline
\end{tabular}

\begin{tablenotes}
      \small
      \item Note: The DGP considered in this table is: $Y_{it} = \beta X_{it} + \alpha_i + \sqrt{ 1+X_{it}^2} \cdot \epsilon_{it}$, $X_{it}\sim i.i.d \text{ } \mathcal{N}(0,1)\cdot \mathbf{1}\{|X_{it}|\leq 2\}$, $\alpha_i,\epsilon_{it}\sim i.i.d \text{ } \mathcal{N}(0,1)$, so $\beta_{\tau}(x) = 1+ \Phi^{-1}(\tau)\cdot x/\sqrt{1+x^2}$.
    \end{tablenotes}
\end{threeparttable}

\end{center}
\label{default}
\end{table}%

\begin{table}[htp]
\small
\begin{center}
\begin{threeparttable}
\caption{Biases and MSEs of the Estimators at $\tau=0.5$ with $T(3)$ Errors.}
\begin{tabular}{cc|cccc|cccc}
\hline
  &  & \multicolumn{4}{c|}{Bias} &  \multicolumn{4}{c}{MSE} \\
  $x$ &   $\beta_{\tau}(x)$  & $\hat{\beta}_{\tau}$ &  $\hat{\beta}_{\tau}^{bc}$  & $\check{\beta}_{\tau}$ &  $\check{\beta}_{\tau}^{bc}$ & $\hat{\beta}_{\tau}$ &  $\hat{\beta}_{\tau}^{bc}$  & $\check{\beta}_{\tau}$ &  $\check{\beta}^{bc}_{\tau}$ \\
\hline
 -2.0&  1.000 & -0.041   &-0.032 & 0.007&-0.012&0.369  &0.587&0.337&0.561 \\
 -1.6&  1.000&   -0.008 &  -0.007& 0.011& 0.015& 0.050 & 0.079&0.052& 0.088\\
  -1.2&   1.000&  -0.005  & -0.004& 0.004&0.005& 0.012 &0.018&0.017& 0.027\\
 -0.8&  1.000 &   -0.004 & -0.002& 0.001 &0.002& 0.006 & 0.009&0.007& 0.010\\
  -0.4&   1.000&   -0.003 & -0.003& -0.000 &0.000& 0.004 &0.005&0.004& 0.006\\
  0.0&  1.000 &   0.002 & 0.002& -0.002& -0.003&  0.003&0.004&0.003& 0.005\\
  0.4&   1.000&    0.002& 0.003 & 0.002&-0.003 & 0.004 &0.005&0.004& 0.006\\
  0.8&  1.000&    -0.005& -0.006& 0.008& 0.007& 0.007 &0.010&0.007& 0.011\\
  1.2&   1.000&    -0.032& -0.002& 0.001&0.002& 0.012 &0.018&0.017& 0.026\\
 1.6&  1.000&  -0.001  & -0.009& -0.012&-0.022& 0.048 & 0.075&0.057& 0.094\\
  2.0&   1.000&   0.016 & -0.007 & -0.012&-0.009& 0.342 &0.565&0.354& 0.580\\
  \hline
\end{tabular}

\begin{tablenotes}
      \small
      \item Note: The DGP considered in this table is: $Y_{it} = \beta X_{it} + \alpha_i + \sqrt{ 1+X_{it}^2} \cdot \epsilon_{it}$, $X_{it}\sim i.i.d \text{ } \mathcal{N}(0,1)\cdot \mathbf{1}\{|X_{it}|\leq 2\}$, $\alpha_i \sim i.i.d \text{ } \mathcal{N}(0,1)$, $\epsilon_{it}\sim T(3)$, so $\beta_{0.5}(x) = 1$.
    \end{tablenotes}
\end{threeparttable}

\end{center}
\label{default}
\end{table}%

\begin{table}[htp]
\small
\begin{center}
\begin{threeparttable}
\caption{Biases and MSEs of the Estimators at $\tau=0.75$ with Gaussian Errors.}
\begin{tabular}{cc|cccc|cccc}
\hline
  &  & \multicolumn{4}{c|}{Bias} &  \multicolumn{4}{c}{MSE} \\
  $x$ &   $\beta_{\tau}(x)$  & $\hat{\beta}_{\tau}$ &  $\hat{\beta}_{\tau}^{bc}$  & $\check{\beta}_{\tau}$ &  $\check{\beta}_{\tau}^{bc}$ & $\hat{\beta}_{\tau}$ &  $\hat{\beta}_{\tau}^{bc}$  & $\check{\beta}_{\tau}$ &  $\check{\beta}^{bc}_{\tau}$ \\
\hline
 -2.0&  0.397 & 0.744   & -0.147 & 0.751& 0.112&0.828  &0.599&0.863&0.518 \\
 -1.6&  0.428&   0.194 &  -0.034& 0.194& 0.027& 0.084 & 0.077&0.088& 0.096\\
  -1.2&   0.482&  0.087  & 0.033& 0.089& 0.041& 0.022 &0.024&0.022& 0.028\\
 -0.8&  0.579 &   0.068 & 0.043&  0.069 & 0.045& 0.011 & 0.011&0.011& 0.014\\
  -0.4&   0.750&   0.054 & 0.042& 0.054 & 0.044& 0.006 &0.007&0.007& 0.008\\
  0.0&  1.000 &   -0.001 & -0.001& -0.001& -0.001&  0.003&0.004&0.003& 0.005\\
  0.4&   1.251&    -0.058& -0.049 & -0.058& -0.048 & 0.007 &0.008&0.008& 0.009\\
  0.8&  1.421&    -0.072&  -0.048& -0.071& -0.046& 0.011 &0.011&0.011& 0.012\\
  1.2&   1.518&    -0.087& -0.035& -0.090& -0.041& 0.022 &0.022&0.023& 0.027\\
 1.6&  1.572&  -0.195  & -0.030& -0.197& -0.024& 0.083 & 0.071&0.087& 0.085\\
  2.0&   1.603&   -0.737 &0.217 & -0.737& -0.053& 0.868 &0.696&0.875& 0.497\\
  \hline
\end{tabular}
\begin{tablenotes}
      \small
      \item Note: The DGP considered in this table is: $Y_{it} = \beta X_{it} + \alpha_i + \sqrt{ 1+X_{it}^2} \cdot \epsilon_{it}$, $X_{it}\sim i.i.d \text{ } \mathcal{N}(0,1)\cdot \mathbf{1}\{|X_{it}|\leq 2\}$, $\alpha_i,\epsilon_{it}\sim i.i.d \text{ } \mathcal{N}(0,1)$, so $\beta_{\tau}(x) = 1+ \Phi^{-1}(\tau)\cdot x/\sqrt{1+x^2}$.
    \end{tablenotes}
\end{threeparttable}
\end{center}
\label{default}
\end{table}%

\begin{table}[htp]
\small
\begin{center}
\begin{threeparttable}
\caption{Biases and MSEs of the Estimators at $\tau=0.75$ with $T(3)$ Errors.}
\begin{tabular}{cc|cccc|cccc}
\hline
  &  & \multicolumn{4}{c|}{Bias} &  \multicolumn{4}{c}{MSE} \\
  $x$ &   $\beta_{\tau}(x)$  & $\hat{\beta}_{\tau}$ &  $\hat{\beta}_{\tau}^{bc}$  & $\check{\beta}_{\tau}$ &  $\check{\beta}_{\tau}^{bc}$ & $\hat{\beta}_{\tau}$ &  $\hat{\beta}_{\tau}^{bc}$  & $\check{\beta}_{\tau}$ &  $\check{\beta}^{bc}_{\tau}$ \\
\hline
 -2.0&  0.316 & 0.936   & -0.222 & 0.960& 0.264&1.311  &0.932&1.408&0.832 \\
 -1.6&  0.351&   0.242 &  0.019& 0.258& 0.058& 0.130 & 0.117&0.139& 0.133\\
  -1.2&   0.412&  0.110  & 0.037& 0.120& 0.044& 0.032 &0.033&0.036& 0.040\\
 -0.8&  0.522 &   0.086 & 0.051&  0.088 & 0.048& 0.017 & 0.016&0.017& 0.018\\
  -0.4&   0.716&   0.065 & 0.049& 0.066 & 0.050& 0.009 &0.010&0.010& 0.012\\
  0.0&  1.000 &   -0.003 & -0.006& -0.000& -0.001&  0.005&0.007&0.005& 0.008\\
  0.4&   1.294&    -0.070& -0.057 & -0.067& -0.055 & 0.010 &0.011&0.010& 0.012\\
  0.8&  1.478&    -0.087&  -0.052& -0.083& -0.047& 0.016 &0.015&0.016& 0.017\\
  1.2&   1.588&    -0.105& -0.034& -0.104& -0.038& 0.032 &0.033&0.034& 0.041\\
 1.6&  1.649&  -0.257  & -0.059& -0.242& -0.030& 0.144 & 0.135&0.136& 0.147\\
  2.0&   1.684&   -0.980 &0.155 & -0.985& -0.302& 1.434 &0.918&1.475& 0.796\\
  \hline
\end{tabular}
\begin{tablenotes}
      \small
      \item Note: The DGP considered in this table is: $Y_{it} = \beta X_{it} + \alpha_i + \sqrt{ 1+X_{it}^2} \cdot \epsilon_{it}$, $X_{it}\sim i.i.d \text{ } \mathcal{N}(0,1)\cdot \mathbf{1}\{|X_{it}|\leq 2\}$, $\alpha_i\sim i.i.d \text{ } \mathcal{N}(0,1)$, $\epsilon_{it}\sim T(3)$, so $\beta_{0.75}(x) = 1+ 0.765\cdot x/\sqrt{1+x^2}$.
    \end{tablenotes}
\end{threeparttable}
\end{center}
\label{default}
\end{table}%


\newpage
\nocite{*}
\begin{spacing}{1.0}
\bibliographystyle{chicago}
\bibliography{llsqr}
\end{spacing}

\end{document}